\documentclass[journal]{IEEEtran}

\usepackage{amsmath,amsthm,amssymb,amsfonts, mathrsfs, microtype, bm}
\usepackage{enumerate}
\usepackage{bbm}
\usepackage{mathtools}

\usepackage[mathscr]{euscript}
\newtheorem{theorem}{Theorem}
\newtheorem{corollary}{Corollary}
\newtheorem{lemma}{Lemma}

\newtheorem{proposition}{Proposition}
\newtheorem{definition}{Definition}

\theoremstyle{remark}
\newtheorem{remark}{Remark}

\usepackage{tikz-cd}
\usepackage{hyperref}
\usepackage[utf8]{inputenc}
\usepackage[square, numbers, comma, sort&compress]{natbib}
\usepackage{tcolorbox}
\usepackage{graphicx}
\usepackage{psfrag}
\usepackage{todonotes}
\usepackage{subcaption}

\usepackage{enumitem}
\usepackage{hyperref}
\hypersetup{colorlinks,linkcolor={blue},citecolor={blue},urlcolor={blue}} 

\newcommand{\mtx}[1]{\mathsf{#1}}
\newcommand{\set}[1]{\mathcal{#1}}

\newcommand{\bfu}{u_1^n}
\newcommand{\bfU}{U_1^n}

\newcommand{\bfX}{X_1^n}

\newcommand{\bfF}{F_1^n}
\newcommand{\bfx}{x_1^n}
\newcommand{\bfy}{y_1^n}

\newcommand{\mle}{\hat{a}_{\text{ML}}(\bfu)}
\newcommand{\MLE}{\hat{a}_{\text{ML}}(\bfU)}
\newcommand{\lambdahat}{\hat{\lambda}_n}

\newcommand{\dis}[2]{\mathsf{d}\left(#1,#2\right)}

\newcommand{\rdf}{\mathbb{R}}

\newcommand{\prob}[1]{\mathbb{P}\left[#1\right]}

\newcommand{\var}[1]{\mathrm{Var}\left[#1\right]}

\newcommand{\EX}[1]{\mathbb{E}\left[#1\right]}

\newcommand{\LRB}[1]{\left(#1\right)}

\newcommand{\dmax}{d_{\mathrm{max}}}

\newcommand{\lrabs}[1]{\left |#1\right |}

\newcommand{\N}{\mathbb{N}}

\newcommand{\R}{\mathbb{R}}

\newcommand{\abs}[1]{\left\vert {#1} \right\vert}
\newcommand{\norm}[1]{\Vert {#1} \Vert}

\newcommand{\pnorm}[2]{\norm{#2}_{#1}}

\newcommand{\lpara}[1]{\left(#1\right)}
\newcommand{\lbrac}[1]{\left[#1\right]}
\newcommand{\lbpara}[1]{\left\{#1\right\}}

\usepackage{stackengine}
\def\delequal{\triangleq}

\usepackage{makecell}

\colorlet{Changes@Color}{red}

\begin{document}
\title{Nonstationary Gauss-Markov Processes:\\ Parameter Estimation and Dispersion}

\author{Peida~Tian,~\IEEEmembership{Student Member,~IEEE,} Victoria~Kostina,~\IEEEmembership{Member,~IEEE}%
\thanks{P. Tian and V. Kostina are with the Department of Electrical Engineering, California Institute of Technology. (e-mail: \hbox{\{ptian, vkostina\}@caltech.edu}). This research was supported in part by the National Science Foundation (NSF) under Grant CCF-1751356.  A preliminary version~\cite{tian2019parameter} of this paper was presented at the 2019 IEEE International Symposium on Information Theory.}
\thanks{Copyright \textsuperscript{\textcopyright}2017 IEEE. Personal use of this material is permitted.  However, permission to use this material for any other purposes must be obtained from the IEEE by sending a request to pubs-permissions@ieee.org.}
}


\maketitle

\begin{abstract}
This paper provides a precise error analysis for the maximum likelihood estimate $\hat{a}_{\text{ML}}(u_1^n)$ of the parameter $a$ given samples $u_1^n = (u_1, \ldots, u_n)'$ drawn from a nonstationary Gauss-Markov process $U_i = a U_{i-1} + Z_i,~i\geq 1$, where $U_0 = 0$, $a> 1$, and $Z_i$'s are independent Gaussian random variables with zero mean and variance $\sigma^2$. We show a tight nonasymptotic exponentially decaying bound on the tail probability of the estimation error. Unlike previous works, our bound is tight already for a sample size of the order of hundreds. We apply the new estimation bound to find the dispersion for lossy compression of nonstationary Gauss-Markov sources. We show that the dispersion is given by the same integral formula that we derived previously for the asymptotically stationary Gauss-Markov sources, i.e., $|a| < 1$. New ideas in the nonstationary case include separately bounding the maximum eigenvalue (which scales exponentially) and the other eigenvalues (which are bounded by constants that depend only on $a$) of the covariance matrix of the source sequence, and new techniques in the derivation of our estimation error bound. 
\end{abstract}

\begin{IEEEkeywords} 
Parameter estimation, maximum likelihood estimator, unstable processes, finite blocklength analysis, lossy compression, sources with memory, rate-distortion theory, system identification, covering in stochastic processes, adaptive control.
\end{IEEEkeywords}

\section{Introduction}

\subsection{Overview}
\IEEEPARstart{W}{e consider} two related problems that concern a scalar Gauss-Markov process $\{U_i\}_{i = 1}^{\infty}$, defined by $U_0 = 0$ and  
\begin{align}
U_i = a U_{i-1} + Z_i, \quad\forall i\geq 1, \label{eqn:GMmodel}
\end{align} 
where $Z_i$'s are independent Gaussian random variables with zero mean and variance $\sigma^2$. 

The first problem is parameter estimation: given samples $u_1^n$ drawn from the Gauss-Markov source, we seek to design and analyse estimators for the unknown system parameter $a$. The consistency and asymptotic distribution of the maximum likelihood (ML) estimator have been studied in the literature~\cite{mann1943statistical, rubin1950consistency, white1958limiting, anderson1959asymptotic, rissanen1979strong, chan1987asymptotic}. Our main contribution is a large deviation bound on the estimation error of the ML estimator. Our numerical experiments indicate that our new bound is tighter than previously known results~\cite{bercu1997large, worms2001large, rantzer2018concentration}. 

The second problem is the nonasymptotic performance of the optimal lossy compressor of the Gauss-Markov process. An encoder outputs $nR$ bits for each realization $u_1^n$. Once the decoder receives the $nR$ bits, it produces $\hat{u}_1^n$ as a reproduction of $u_1^n$. The distortion between $u_1^n$ and $\hat{u}_1^n$ is measured by the mean squared error (MSE). Two commonly used criteria to quantify the distortion of a lossy compression scheme are the average distortion criterion and the excess-distortion probability criterion. The rate-distortion theory, initiated by Shannon~\cite{shannon1959coding} and further pioneered in~\cite{goblick1969coding, kolmogorov1956shannon, berger1968rate, berger1970information, berger1971rate, gray1970information, gray1971markov, gray2008note, wyner1971bounds,  marton1974error, hashimoto1980rate}, studies the optimal tradeoff between the rate $R$ and the distortion. In the limit of large blocklength $n$, the minimum rate $R$ required to achieve average distortion $d$ is given by the rate-distortion function. The nonasymptotic version of the rate-distortion problem~\cite{marton1974error, zhang1997redundancy, yang1999redundancy, ingber2011dispersion, kostina2012fixed} studies the rate-distortion  tradeoff for finite blocklength $n$. Our main contribution is a coding theorem that characterizes the gap between the rate-distortion function and the minimum rate $R$ at blocklength $n$ for the nonstationary Gauss-Markov source ($a>1$), under the excess-distortion probability criterion. We leverage our result on the ML estimator to analyze lossy compression. Namely, we apply our bound on the estimation error of the ML estimator to construct a typical set of the sequences whose estimated parameter $a$ is close to the true $a$. We then use the typical set in our achievability proof of the nonasymptotic coding theorem.

Without loss of of generality, we assume that $a\geq 0$ in this paper, since, otherwise, we can consider another random process $\{U'_i\}_{i = 1}^{\infty}$ defined by the invertible mapping $U'_i \delequal (-1)^{i} U_i$ that satisfies $ U'_i = (-a) U'_{i-1} + (-1)^i Z_{i}$, where $(-1)^i Z_{i}$'s are also independent zero-mean Gaussian random variables with variance $\sigma^2$. We distinguish the following three cases: 
\begin{itemize}
\item $0< a <1$: the asymptotically stationary case;
\item $a=1$: the unit-root case; 
\item $a>1$:  the nonstationary case.
\end{itemize}
In this paper, we mostly focus on the nonstationary case.

\subsection{Motivations}
Estimation of parameters of stochastic processes from their realizations has many applications. In the statistical analysis of economic time series~\cite{mann1943statistical, haavelmo1943statistical, koopmans1942serial}, the Gauss-Markov process $\{U_i\}_{i=1}^{\infty}$ is used to model the varying price of a certain commodity with time, and the ML estimate of the unknown coefficient $a$ is then used to predict future prices. In~\cite{gould1974stochastic} and~\cite[Sec. 5]{dickey1979distribution}, the Gauss-Markov process with $a = 1$ is used to model the stochastic structure of the velocity of money. The Gauss-Markov process, also known as the autoregressive process of order 1 (AR(1)), is a special case of the general autoregressive-moving-average (ARMA) model~\cite{whittle1951hypothesis, box1970time}, for which various estimation and prediction procedures have been proposed, e.g. the Box-Jenkins method~\cite{box1970time}. The Gauss-Markov process is also a special case of the linear state-space model (e.g.~\cite[Chap. 5]{kailath2000linear}) that is popular in control theory. One of the problems in control is system identification~\cite{ljung1987system}, which is the problem of building mathematical models using measured data from unknown dynamical systems. Parameter estimation is one of the common methods used in system identification where the dynamical system is modeled by a state-space model~\cite[Chap. 7]{ljung1987system} with unknown parameters. In modern data-driven control systems, where the goal is to control an unknown nonstationary system given measured data, parameter estimation methods are used as a first step in designing controllers~\cite{rantzer2018concentration}~\cite[Sec. 1.2]{tu2019sample}. In speech signal processing, the linear predictive coding algorithm~\cite{atal1971speech} relies on parameter estimation (the ordinary least squares estimate, or, equivalently, the maximum likelihood estimate assuming Gaussian noise) to fit a higher-order Gauss-Markov process, see~\cite[App. C]{atal1971speech}. A fine-grained analysis of the ML estimate is instrumental in optimizing the design of all these systems. Our nonasymptotic analysis leading up to a large deviation bound for the ML estimate in our simple setting can provide insights for analyzing more complex random processes, e.g., higher-order autoregressive processes and vector systems. 

Understanding finite-blocklength lossy compression of the Gauss-Markov process fits into a continuing effort by many researchers to advance the rate-distortion theory of information sources with memory, see~\cite{kolmogorov1956shannon, berger1968rate, berger1970information, gray1970information, gray1971markov, wyner1971bounds, hashimoto1980rate, kontoyiannis2000pointwise, dembo2002source, kontoyiannis2003pattern, kontoyiannis2006mismatched, venkataramanan2007source, gray2008note, kontoyiannis2002arbitrary, dembo1999asymptotics, madiman2004minimum}, as well as into a newer push~\cite{marton1974error, zhang1997redundancy, yang1999redundancy, ingber2011dispersion, kostina2012fixed, kostina2013lossy, tan2014dispersions, watanabe2017second, dispersionJournal, zhou2016discrete, zhou2017second} to understand the fundamental limits of low latency communication. There is a tight connection between lossy compression of the nonstationary Gauss-Markov process and control of an unstable linear system under communication constraints~\cite{tatikonda2004stochastic, kostina2019rate}. Namely, the minimum channel capacity needed to achieve a given LQG (linear quadratic Gaussian) cost for the plant~\cite[Eq. (1)]{tatikonda2004stochastic} is lower-bounded by the causal rate-distortion function of the Gauss-Markov process~\cite[Eq. (9)]{tatikonda2004stochastic}. See~\cite[Th. 1]{kostina2019rate} for more details. Being more restrictive on the coding schemes, the causal rate-distortion function is further lower-bounded by the traditional rate-distortion function. The result in this paper on the rate-distortion tradeoff in the finite blocklength regime provides a lower bound on the minimum communication rate required to ensure that the LQG cost stays below a desired threshold with desired probability at the end of a finite horizon. Finally, the aforementioned linear predictive coding algorithm~\cite{atal1971speech} is connected to lossy compression of autoregressive processes, see a recent historical note by Gray~\cite[p.2]{gray2020in}.

\subsection{Notations}
For $n\in \mathbb{N}$, we use $[n]$ to denote the set $\{1, 2, ..., n\}$. We use the standard notations for the asymptotic behaviors $O(\cdot), o(\cdot)$, $\Theta(\cdot)$, $\Omega(\cdot)$ and $\omega(\cdot)$. Namely, let $f(n)$ and $g(n)$ be two functions of $n$, then $f(n) =O(g(n))$ means that there exists a constant $c>0$ and $n_0\in\mathbb{N}$ such that $|f(n)|\leq c |g(n)|$ for any $n\geq n_0$; $f(n) = o(g(n))$ means $\lim_{n\rightarrow\infty} f(n) /g(n) = 0$; $f(n) = \Theta(g(n))$ means there exist positive constants $c_1, c_2$ and $n_0\in\mathbb{N}$ such that $c_1 g(n) \leq f(n) \leq c_2 g(n)$ for any $n\geq n_0$; $f(n) = \Omega(g(n))$ if and only if $g(n) = O(f(n))$; and $f(n) = \omega(g(n))$ if and only if $\lim_{n\rightarrow\infty} f(n) / g(n) = +\infty$. For a matrix $\mtx{M}$, we denote by $\mtx{M}'$ its transpose, by $\|\mtx{M}\|$ its operator norm (the largest singular value) and by $\mu_1(\mtx{M}) \leq \ldots \leq \mu_n(\mtx{M})$ its eigenvalues listed in nondecreasing order. We use $\set{S}^c$ to denote the complement of a set $\set{S}$. All logarithms and exponentials are base $e$.

\section{Previous Works}
\label{sec:PF}
\subsection{Parameter Estimation}
The maximum likelihood (ML) estimate $\mle$ of the parameter $a$ given samples $\bfu= (u_1, \ldots, u_n)'$ drawn from the Gauss-Markov source is given by
\begin{align}
\hat{a}_{\text{ML}}(\bfu)= \frac{\sum_{i = 1}^{n-1} u_{i}u_{i+1}}{\sum_{i = 1}^{n-1} u_{i}^2}. \label{eqn:MLEintro}
\end{align}
The derivation of~\eqref{eqn:MLEintro} is straightforward, e.g.~\cite[App. F-A]{dispersionJournal}. The problem is to provide performance guarantees of $\mle$. This simply formulated problem has been widely studied in the literature. Our main contribution in this paper is a nonasymptotic fine-grained large deviations analysis of the estimation error. 

The estimate $\mle$ in~\eqref{eqn:MLEintro} has been extensively studied in the statistics~\cite{white1958limiting, rissanen1979strong} and economics~\cite{mann1943statistical, rubin1950consistency} communities. Mann and Wald~\cite{mann1943statistical} and Rubin~\cite{rubin1950consistency} showed that the estimation error $\MLE - a$ converges to 0 in probability for any $a\in \mathbb{R}$. Rissanen and Caines~\cite{rissanen1979strong} later proved that $\MLE - a$ converges to 0 almost surely for $0< a<1$. To better understand the finer scaling of  the error $\MLE-a$, researchers turned to study the limiting distribution of the normalized estimation error $h(n)(\MLE - a)$ for a careful choice of the standardizing function $h(n)$: 
\begin{align}
h(n) \delequal \begin{cases}
\sqrt{\frac{n}{1 - a^2}}, & |a| < 1,\\
\frac{n}{\sqrt{2}}, & |a| = 1, \\
\frac{|a|^n}{a^2 - 1}, & |a| > 1.
\end{cases}
\end{align}
With the above choices of $h(n)$, Mann and Wald~\cite{mann1943statistical} and White~\cite{white1958limiting} showed that the distribution of the normalized estimation error $h(n)(\MLE- a)$ converges to $\mathcal{N}(0, 1)$ for $|a|<1$; to the standard Cauchy distribution for $|a|>1$; and for $|a|=1$, to the distribution of 
\begin{align}
\frac{B^2(1) - 1}{2\int_{0}^1 B^2(t)~dt},
\end{align}
where $\{B(t):t\in [0,1]\}$ is a Brownian motion.

Generalizations of the above results in several directions have also been investigated. In~\cite[Sec. 4]{mann1943statistical}, the maximum likelihood estimator for the $p$-th order stationary autoregressive processes with $Z_i$'s being i.i.d. zero-mean and bounded moments random variables (not necessarily Gaussian) was shown to be weakly consistent, and the scaled estimation errors $\sqrt{n}(\hat{a}_j - a_j)$ for $j = 1, \ldots, p$ were shown to converge in distribution to the Gaussian random variables as $n$ tends to infinity. Anderson~\cite[Sec. 3]{anderson1959asymptotic} studied the limiting distribution of the maximum likelihood estimator for a nonstationary vector version of the process~\eqref{eqn:GMmodel}. Chan and Wei~\cite{chan1987asymptotic} studied the performance of the estimation error when $a$ is not a constant but approaches to 1 from below in the order of $1/n$. Estimating $a$ from a block of outcomes of the Gauss-Markov source~\eqref{eqn:GMmodel} is one of the simplest versions of the problem of system identification, where the goal is to learn system parameters of a dynamical system from the observations~\cite{simchowitz2018learning, oymak2018non, sarkar2018fast, faradonbeh2018finite, rantzer2018concentration}. One objective of those studies is to obtain tight performance bounds on the least-squares estimates of the system parameters $\mtx{A}, \mtx{B}, \mtx{C}, \mtx{D}$ from a single input / output trajectory $\{W_i, Y_i\}_{i=1}^{n}$ in the following state-space model, e.g.~\cite[Eq. (1)--(2)]{oymak2018non}:
\begin{align}
X_{i+1} &= \mtx{A}X_i + \mtx{B}W_{i} + Z_i, \\
Y_i &= \mtx{C}X_i + \mtx{D}W_i +V_i,
\end{align}
where $X_i, W_i, Z_i,V_i$'s are random vectors of certain dimensions and the system parameters $\mtx{A}, \mtx{B}, \mtx{C}, \mtx{D}$ are matrices of appropriate dimensions. The Gauss-Markov process in~\eqref{eqn:GMmodel} can be written as the state-space model by choosing $\mtx{A} = a$ being a scalar, $\mtx{B} = \mtx{D} = 0$, $\mtx{C} = 1$ and $V_i = 0$. For stable vector systems, that is, $\|\mtx{A}\| < 1$, Oymak and Ozay~\cite[Thm. 3.1]{oymak2018non} showed that the estimation error in spectral norm is $O(1/\sqrt{n})$ with high probability, where $n$ is the number of samples. For the subclass of the regular unstable systems~\cite[Def. 3]{faradonbeh2018finite}, Faradonbeh et al.~\cite[Thm. 1]{faradonbeh2018finite} proved that the probability of estimation error exceeding a positive threshold in spectral norm decays exponentially in $n$. For the Gauss-Markov processes considered in the present paper, Simchowitz et al.~\cite[Thm. B.1]{simchowitz2018learning} and Sarkar and Rakhlin~\cite[Prop. 4.1]{sarkar2018fast} presented tail bounds on the estimation error of the ML estimate.

Another line of work closely related to this paper is the large deviation principle (LDP)~\cite[Ch. 1.2]{dembo1994zeitouni} on $\MLE - a$. Given an error threshold $\eta > 0$, define $P^+(n, a, \eta)$ and $P^-(n, a, \eta)$ as follows:
\begin{align}
P^+(n, a, \eta) &\delequal -\frac{1}{n}\log\prob{\MLE - a > \eta},\label{def:pplus}\\
P^-(n, a, \eta) &\delequal -\frac{1}{n}\log\prob{\MLE- a  <  -\eta}.\label{def:pminus}
\end{align}
We also define $P(n, a, \eta)$ as
\begin{align}
P(n, a, \eta) \triangleq -\frac{1}{n}\log \prob{ |\MLE - a | > \eta}.\label{def:p}
\end{align}
The large deviation theory studies the rate functions, defined as the limits of $P^+(n, a, \eta)$, $P^-(n, a, \eta)$ and $P(n, a, \eta)$, as $n$ goes to infinity. Bercu et al.~\cite[Prop. 8]{bercu1997large} found the rate function for the case of $0<a<1$.  For $a\geq 1$, Worms~\cite[Thm. 1]{worms2001large} proved that the rate functions can be bounded from below  implicitly by the optimal value of an optimization problem.

These studies of the limiting distribution and the LDP of the estimation error are asymptotic. In this paper, we develop a nonasymptotic analysis of the estimation error. Two nonasymptotic lower bounds on $P^+(n, a, \eta)$ and $P^-(n, a, \eta)$ are available in the literature. For any $a\in\mathbb{R}$, Rantzer~\cite[Th. 4]{rantzer2018concentration} showed that
\begin{align}
P^+(n, a, \eta)~~\left (\text{and }P^-(n, a, \eta)\right )~\geq \frac{1}{2}\log (1 + \eta^2).\label{eqn:rantzer}
\end{align}
Bercu and Touati~\cite[Cor. 5.2]{bercu2008exponential} proved that
\begin{align}
P^+(n, a, \eta)~~\left (\text{and }P^-(n, a, \eta)\right )~\geq  \frac{\eta^2}{2(1 + y_\eta)},\label{eqn:bercu}
\end{align}
where $y_\eta$ is the unique positive solution to $(1+x)\log (1 + x)-x -\eta^2 = 0$ in $x$. Both bounds~\eqref{eqn:rantzer} and~\eqref{eqn:bercu} do not capture the dependence on $a$ and $n$, and are the same for $P^+(n, a, \eta)$ and $P^-(n, a, \eta)$. The bounds in~\cite{simchowitz2018learning, oymak2018non, sarkar2018fast, faradonbeh2018finite, rantzer2018concentration} either are optimal only order-wise or involve implicit constants. Our main result on parameter estimation is a tight nonasymptotic lower bound on $P^+(n, a, \eta)$ and $P^-(n, a, \eta)$. For larger $a$, the lower bound becomes larger, which suggests that unstable systems are easier to estimate than stable ones, an observation consistent with~\cite{simchowitz2018learning}. The proof is inspired by Rantzer~\cite[Lem. 5]{rantzer2018concentration}, but our result improves Rantzer's result~\eqref{eqn:rantzer} and Bercu and Touati's result~\eqref{eqn:bercu}, see Fig.~\ref{fig:compare} for a comparison. Most of our results generalize to the case where $Z_i$'s are i.i.d. sub-Gaussian random variables, see Theorem~\ref{thm:subgaussian} in Section~\ref{subsec:gen2sub} below.

\subsection{Nonasymptotic Rate-distortion Theory}
\label{subsec:nrdt}
The rate-distortion theory studies the problem of compressing a generic random process $\{X_i\}_{i=1}^{\infty}$ with minimum distortion. Given a distortion threshold $d>0$, an excess-distortion probability $\epsilon \in (0,1)$ and the number of codewords $M\in \mathbb{N}$, an $(n, M, d, \epsilon)$ lossy compression code for a random vector $X_1^n$ consists of an encoder $\mathsf{f}_n \colon \mathbb{R}^n \rightarrow [M]$, and a decoder $\mathsf{g}_n\colon [M] \rightarrow \mathbb{R}^n$, such that $\mathbb{P}\left[\mathsf{d}\left(X_1^n, \mathsf{g}_n\left(\mathsf{f}_n(X_1^n)\right)\right) > d\right] \leq \epsilon$, where $\mathsf{d}(\cdot, \cdot)$ is the distortion measure. This paper considers the mean squared error (MSE) distortion: $\forall~ \bfx,~\bfy\in \mathbb{R}^n$, 
\begin{align}
\mathsf{d}(\bfx, \bfy) \delequal \frac{1}{n}\sum_{i = 1}^n (x_i - y_i)^2. 
\label{eqn:disdef}
\end{align}
The minimum achievable code size and source coding rate are defined respectively by 
\begin{align}
M^\star (n, d, \epsilon) &\delequal \min\left\{M\in \mathbb{N} \colon \exists~ (n, M, d, \epsilon) \text{ code}\right\}, \\
R(n, d, \epsilon) &\delequal \frac{1}{n} \log M^\star (n, d, \epsilon). \label{def:Rnde}
\end{align}
In this paper, we approximate the nonasymptotic coding rate $R(n, d, \epsilon)$ for the nonstationary Gauss-Markov source.

Another related and widely studied setting is compression under the average distortion criterion. Given a distortion threshold $d>0$ and the number of codewords $M\in \mathbb{N}$, an $(n, M, d)$ lossy compression code for a random vector $X_1^n$ consists of an encoder $\mathsf{f}_n \colon \mathbb{R}^n \rightarrow [M]$, and a decoder $\mathsf{g}_n\colon [M] \rightarrow \mathbb{R}^n$, such that $\mathbb{E}\left[\mathsf{d}\left(X_1^n, \mathsf{g}_n\left(\mathsf{f}_n(X_1^n)\right)\right) \right] \leq d$. Similarly, one can define $M^\star(n, d)$ and $R(n, d)$ as the minimum achievable code size and source coding rate, respectively, under the average distortion criterion. The traditional rate-distortion theory~\cite{shannon1959coding, goblick1969coding,gray1970information,gray1971markov,berger1970information,berger1971rate} showed that the limit of the operational source coding rate $R(n, d)$ as $n$ tends to infinity equals the informational rate-distortion function for a wide class of sources. For discrete memoryless sources, Zhang, Yang and Wei in~\cite{zhang1997redundancy} showed that $R(n, d)$ approaches the rate-distortion function as $\log n / 2n + o(\log n / n )$. For abstract alphabet memoryless sources, Yang and Zhang in~\cite[Th. 2]{yang1999redundancy} showed a similar convergence rate.

Under the excess-distortion probability criterion, one can also study the nonasymptotic behavior of the minimum achievable excess-distortion probability $\epsilon^\star (n, d, M)$: 
\begin{align}
\epsilon^\star (n, d, M) &\delequal \inf\left\{\epsilon > 0\colon \exists~ (n, M, d, \epsilon) \text{ code}\right\}. \label{def:Pndm}
\end{align}
Marton's excess distortion exponent~\cite[Th. 1, Eq. (2)-(3), (20)]{marton1974error} showed that for discrete memoryless sources $P_{X}$, it holds that
\begin{align}
-\frac{1}{n}\log \epsilon^\star (n, d, M) = \min_{P_{\hat{X}}}~D( P_{\hat{X}} || P_{X}) + O\left(\frac{\log n}{n}\right), \label{eqn:marton}
\end{align}
where the minimization is over all probability distributions $P_{\hat{X}}$ such that $\mathbb{R}_{\hat{X}}(d)\geq \frac{\log M}{n}$, where $M$ is such that $\frac{\log M}{n}$ is a constant, $\mathbb{R}_{\hat{X}}(d)$ denotes the rate-distortion function of a discrete memoryless source with single-letter distribution $P_{\hat{X}}$, and $D(\cdot||\cdot)$ denotes the Kullback-Leibler divergence. As pointed out by~\cite[p.~2]{ingber2011dispersion}, for fixed $d>0$ and $\epsilon\in (0,1)$, even the limit of $R(n,d, \epsilon)$ as $n$ goes to infinity is unanswered by Marton's bound in~\eqref{eqn:marton}. Ingber and Kochman~\cite{ingber2011dispersion} (for finite-alphabet and Gaussian sources) and Kostina and Verd{\'u}~\cite{kostina2012fixed} (for abstract sources) showed that the minimum achievable source coding rate $R(n,d,\epsilon)$ admits the following expansion, known as Gaussian approximation~\cite{polyanskiy2010channel}. 
\begin{align}
R(n, d, \epsilon) = \mathbb{R}_{X}(d) + Q^{-1}(\epsilon)\sqrt{\frac{\mathbb{V}(d)}{n}} + O\left(\frac{\log n}{n}\right), \label{eq:gauapp}
\end{align}
where $\mathbb{V}(d)$ is the dispersion of the source (defined as the variance of the tilted information random variable, details later) and $Q^{-1}$ denotes the inverse Q-function. In this paper, by extending our previous analysis~\cite[Th. 1]{dispersionJournal} of the stationary Gauss-Markov source to the nonstationary one, we establish the  Gaussian approximation in the form of~\eqref{eq:gauapp} for the nonstationary Gauss-Markov sources. One of the key ideas behind this extension is to construct a typical set using the ML estimate of $a$, and to use our estimation error bound to probabilistically characterize that set. 
\section{Parameter Estimation}
\label{sec:mainresults}

\subsection{Nonasymptotic Lower Bounds}
\label{subsec:pa}

We first present our nonasymptotic bounds on $P^+(n, a, \eta)$ and $P^-(n, a, \eta)$, defined in~\eqref{def:pplus} and~\eqref{def:pminus} above, respectively. We define two sequences $\{\alpha_\ell\}_{\ell\in\mathbb{N}}$ and $\{\beta_\ell\}_{\ell\in\mathbb{N}}$ as follows. Let $\sigma^2 > 0$ and $a>1$ be fixed constants. For $\eta>0$ and a parameter $s>0$, let $\alpha_\ell$ be the following sequence
\begin{align}
\alpha_1 &\delequal \frac{\sigma^2s^2 - 2\eta s}{2}, \label{alpha1}\\ 
\alpha_\ell & = \frac{\left [ a^2 + 2\sigma^2s(a+\eta)\right ]\alpha_{\ell - 1} + \alpha_1}{1 - 2\sigma^2\alpha_{\ell-1}},\quad \forall \ell\geq 2. \label{alphaEll}
\end{align}
Similarly, let $\beta_\ell$ be the following sequence
\begin{align}
\beta_1 &\delequal \frac{\sigma^2s^2 - 2\eta s}{2}, \label{beta1}\\ 
\beta_\ell & = \frac{\left [a^2 + 2\sigma^2s(-a+\eta)\right ]\beta_{\ell - 1} + \beta_1}{1 - 2\sigma^2\beta_{\ell-1}},\quad \forall \ell\geq 2.\label{betaEll}
\end{align}
Note the subtle difference between~\eqref{alphaEll} and~\eqref{betaEll}: there is a negative sign in the numerator in \eqref{betaEll}. Both sequences depend on $\eta$ and $s$. We derive closed-form expressions and analyze the convergence properties of $\alpha_{\ell}$ and $\beta_{\ell}$ in Appendices~\ref{app:seqA} and~\ref{app:seqB} below. For $\eta>0$ and $n\in\mathbb{N}$, we define the following sets
\begin{align}
\mathcal{S}_n^+ \delequal \left\{s\in\mathbb{R}\colon s>0,~\alpha_{\ell} < \frac{1}{2\sigma^2},~\forall\ell\in [n]\right\}, \label{Splus}\\
\mathcal{S}_n^- \delequal \left\{s\in\mathbb{R}\colon s>0,~\beta_{\ell} < \frac{1}{2\sigma^2},~\forall\ell\in [n]\right\}.\label{Sminus}
\end{align}

\begin{theorem}
\label{thm:chernoff}
For any constant $\eta > 0$, the estimator~\eqref{eqn:MLEintro} satisfies for any $n\geq 2$,
\begin{align}
P^+(n, a, \eta) &\geq \sup_{s \in \mathcal{S}_n^+}~ \frac{1}{2n}\sum_{\ell = 1}^{n-1} \log\LRB{1 - 2\sigma^2\alpha_\ell},
\label{eqn:chernoffUpper}\\
P^-(n, a, \eta) &\geq \sup_{s \in\mathcal{S}_n^-}~ \frac{1}{2n}\sum_{\ell = 1}^{n-1} \log\LRB{1 - 2\sigma^2\beta_\ell},
\label{eqn:chernoffLower}
\end{align}
where $\alpha_{\ell}$ and $\beta_{\ell}$ are defined in~\eqref{alphaEll} and~\eqref{betaEll}, respectively, and $\mathcal{S}_n^+$ and $\mathcal{S}_n^-$ are defined in~\eqref{Splus} and~\eqref{Sminus}, respectively. 
\end{theorem}

Theorem~\ref{thm:chernoff} is a useful result for numerically computing lower bounds on $P^+(n, a, \eta)$ and $P^-(n, a, \eta)$. In Fig.~\ref{fig:compare}, we plot our lower bounds in Theorem~\ref{thm:chernoff}, previous results in~\eqref{eqn:rantzer} by Rantzer and~\eqref{eqn:bercu} by Bercu and Touati, and a simulation result. As one can see, our bound in Theorem~\ref{thm:chernoff} is much tighter than previous results.

The proof of Theorem~\ref{thm:chernoff}, presented in Appendix~\ref{app:pfThmChernoff} below, is a detailed analysis of the Chernoff bound using the tower property of conditional expectations. The proof is motivated by~\cite[Lem. 5]{rantzer2018concentration}, but our analysis is more accurate and the result is significantly tighter, see Fig.~\ref{fig:compare} and Fig.~\ref{fig:compareLim} for comparisons. One recovers Rantzer's lower bound~\eqref{eqn:rantzer} by setting $s = \eta / \sigma^2$ and bounding $\alpha_{\ell}$ as $\alpha_{\ell} \leq \alpha_1$ (due to the monotonicity of $\alpha_{\ell}$ shown in Appendix~\ref{app:seqA} below) in Theorem~\ref{thm:chernoff}. We explicitly state where we diverge from \cite[Lem. 5]{rantzer2018concentration} in the proof in Appendix~\ref{app:pfThmChernoff} below.

\begin{remark}
In view of the G{\"a}rtner-Ellis theorem~\cite[Th. 2.3.6]{dembo1994zeitouni}, we conjecture that the bounds~\eqref{eqn:chernoffUpper} and~\eqref{eqn:chernoffLower} can be reversed in the limit of large $n$:
\begin{align}
\limsup_{n\to\infty} P^+(n, a, \eta) \leq \limsup_{n\to\infty}\sup_{s \in \mathcal{S}_n^+}~ \frac{1}{2n}\sum_{\ell = 1}^{n-1} \log\LRB{1 - 2\sigma^2\alpha_\ell},
\end{align}
and similarly for~\eqref{eqn:chernoffLower}.
\end{remark}

\begin{figure}[!htb]
\centering
\psfrag{a}[lc][][0.8][0]{simulation}
\psfrag{b}[lc][][0.8][0]{Bercu and Touati~\eqref{eqn:bercu}}
\psfrag{c}[lc][][0.8][0]{Rantzer~\eqref{eqn:rantzer}}
\psfrag{d}[lc][][0.8][0]{\eqref{eqn:chernoffUpper} in Theorem~\ref{thm:chernoff}}
\psfrag{e}[lc][][0.8][0]{\eqref{eqn:PplusL} in Theorem~\ref{thm:upperLDP}}
\includegraphics[width=0.48\textwidth]{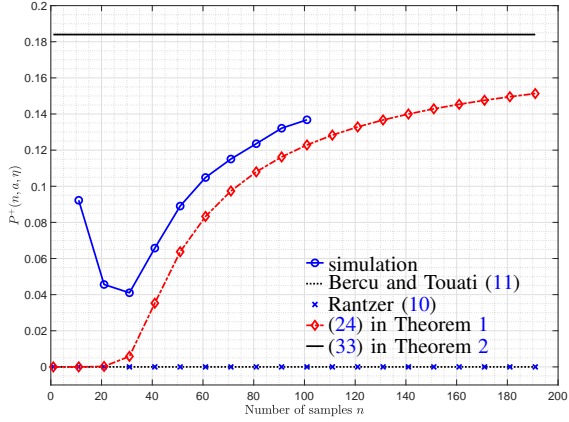}
\caption{Numerical simulations and lower bounds on $P^+(n, a, \eta)$. We choose $a = 1.2$ and $\eta = 10^{-3}$. The ``simulation" curve is obtained as follows. For each $n$, we generate $N = 10^6$ independent samples $u_1^n$ from the Gauss-Markov process~\eqref{eqn:GMmodel}. We approximate $P^+(n, a, \eta)$ by $-\frac{1}{n}\log \left(\frac{1}{N}\#\left\{\text{samples }u_1^n\text{ with } \mle - a > \eta\right\}\right)$, which is shown by the ``simulation" curve. }
\label{fig:compare}
\end{figure}

\subsection{Asymptotic Lower Bounds}
\label{subsec:apa}
We next present our bounds on the error exponents, that is, the limits of $P^{+}(n, a, \eta)$, $P^{-}(n, a, \eta)$ and $P(n, a, \eta)$ as $n$ tends to infinity. To take limits using~\eqref{eqn:chernoffUpper} and~\eqref{eqn:chernoffLower}, we need to understand the two sequences of sets $\mathcal{S}_n^+$ and $\mathcal{S}_n^-$. Define the limits of the sets as 
\begin{align}
\mathcal{S}_\infty^{+} &\delequal \bigcap_{n\geq 1}\mathcal{S}_n^+,\\
\mathcal{S}_\infty^{-} &\delequal \bigcap_{n\geq 1}\mathcal{S}_n^-.
\end{align}
We have the following properties.

\begin{lemma}
\label{lemma:limiting}
Fix any constant $\eta>0$. 
\begin{itemize}
\item (Monotone decreasing sets)  For  any $n\geq 1$, we have
\begin{align}
\mathcal{S}_{n+1}^{+} \subseteq \mathcal{S}_{n}^{+},\quad \mathcal{S}_{n+1}^{-} \subseteq \mathcal{S}_{n}^{-} . 
\end{align}

\item (Limits of the sets) It holds that 
\begin{align}
\mathcal{S}_\infty^{+} =  \left(0,~\frac{2\eta}{\sigma^2}\right], \label{plus}
\end{align}
\begin{align}
 \mathcal{S}_\infty^{-} \supsetneqq\left(0,~\frac{2\eta}{\sigma^2}\right]. \label{minus}
\end{align}
\end{itemize}
\end{lemma}

The proof of Lemma~\ref{lemma:limiting} is presented in Appendix~\ref{app:pfLemLimiting} below. The exact characterization of $\mathcal{S}_{n}^+$ and $\mathcal{S}_n^-$ for each $n$ using $\eta$ is involved. One can see from the definitions~\eqref{Splus} and~\eqref{Sminus} that 
\begin{align}
\mathcal{S}_1^+ = \mathcal{S}_1^- = \left\{s\in\mathbb{R}\colon 0 < s < \frac{\eta + \sqrt{1 +\eta^2}}{\sigma^2}\right\}.
\end{align}
To obtain the set $\mathcal{S}_{n+1}^+$ from $\mathcal{S}_{n}^+$, we need to solve $\alpha_{n+1} < \frac{1}{2\sigma^2}$, which is equivalent to solving an additional inequality involving a polynomial of degree $n+2$ in $s$ (using the closed-form expression for $\alpha_{n+1}$ in~\eqref{eqn:expAlpha} in Appendix~\ref{app:seqA} below). Fig.~\ref{fig:set} presents a plot of $\mathcal{S}_n^+$ for $n= 1, ..., 5$. Despite the complexity of the sets $\mathcal{S}_n^+$ and $\mathcal{S}_n^-$, Lemma~\ref{lemma:limiting} shows their monotonicity property and limits.

\begin{figure}[!htb]
\centering
\includegraphics[width=0.48\textwidth]{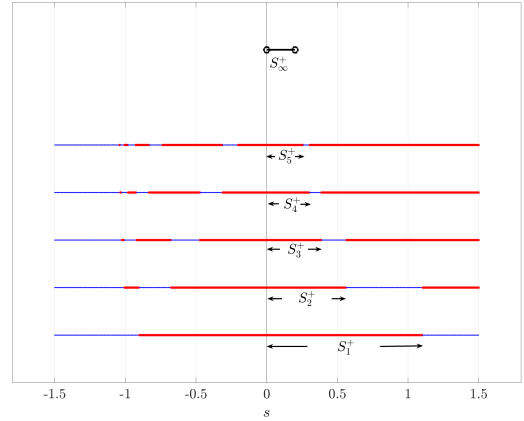}
\caption{Numerical computation of the sets $\mathcal{S}_n^+$ for $a = 1.2$ and $\eta = 0.1$. Each horizontal line corresponds to $n= 1, ..., 5$ in the bottom-up order. Within each horizontal line, the red thick parts denote the ranges of $s$ for which $\alpha_{n} < \frac{1}{2\sigma^2}$, and the blue thin region is where $\alpha_{n} \geq \frac{1}{2\sigma^2}$. The plot for $\mathcal{S}_n^-$ is similar.}
\label{fig:set}
\end{figure}

Combining Theorem~\ref{thm:chernoff} and Lemma~\ref{lemma:limiting}, we obtain the following lower bounds on the error exponents. The proof is given in Appendix~\ref{app:pfupperLDP} below. 
\begin{theorem}
\label{thm:upperLDP}
Fix any constant $\eta>0$. For the ML estimator~\eqref{eqn:MLEintro}, the following three inequalities hold:
\begin{align}
\liminf_{n\rightarrow \infty}~P^+(n, a, \eta) & \geq I^+(a, \eta) \delequal \log (a + 2\eta ), \label{eqn:PplusL} \\
\liminf_{n\rightarrow \infty}~P^-(n, a, \eta) &\geq I^-(a, \eta), \label{eqn:PminusL} \\
\liminf_{n\rightarrow \infty}~P(n, a, \eta) &\geq I^-(a, \eta), \label{eqn:PL}
\end{align}
where 
\begin{align}
I^-(a, \eta)  \delequal
\begin{cases} 
\log a, & 0< \eta \leq \eta_1,\\
\frac{1}{2}\log \frac{2a\eta - (a^2 - 1)}{1-(\eta-a)^2},  & \eta_1 < \eta < \eta_2, \\
 \log(2\eta - a),& \eta  \geq \eta_2,
\end{cases}
\end{align}
with the thresholds $\eta_1$ and $\eta_2$ given by 
\begin{align}
\eta_1 &\triangleq \frac{a^2 - 1}{a}, \label{eta1}\\
\eta_2 &\triangleq \frac{3a + \sqrt{a^2+8}}{4}.\label{eta2}
\end{align}
\end{theorem}

\begin{remark}
The results in~\eqref{plus}-\eqref{minus} and~\eqref{eqn:PplusL}-\eqref{eqn:PminusL} indicate the asymmetry between $P^+(n, a, \eta)$ and $P^-(n, a, \eta)$: the set $\mathcal{S}_{\infty}^-$ has a larger range than $\mathcal{S}_{\infty}^+$, and $I^+(a,\eta) > I^-(a,\eta)$, which suggests that the maximum likelihood estimator $\MLE$ is more likely to underestimate $a$ than to overestimate it. 
\end{remark}

Fig.~\ref{fig:compareLim} presents a comparison of~\eqref{eqn:PL}, Rantzer's bound~\eqref{eqn:rantzer} and Bercu and Touati~\eqref{eqn:bercu}. Our bound~\eqref{eqn:PL} is tighter than both of them for any $\eta>0$. 


\begin{figure}[!htb]
\centering
\psfrag{a}[lc][][0.8][0]{Rantzer~\eqref{eqn:rantzer}}
\psfrag{b}[lc][][0.8][0]{Bercu and Touati~\eqref{eqn:bercu}}
\psfrag{c}[lc][][0.8][0]{\eqref{eqn:PL} in Theorem~\ref{thm:upperLDP}}
\includegraphics[width=0.48\textwidth]{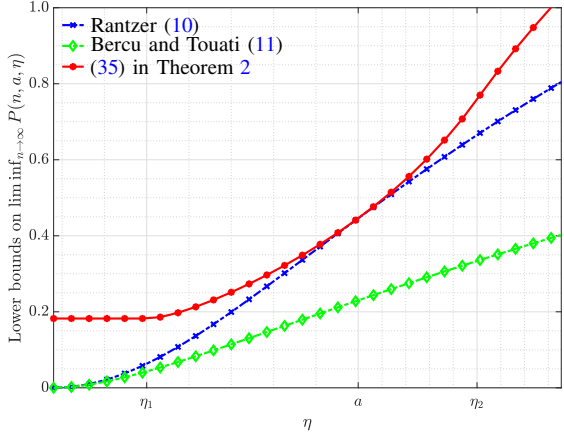}
\caption{Lower bounds on $\liminf_{n\to\infty}P(n, a, \eta)$ for $a = 1.2$.}
\label{fig:compareLim}
\end{figure}

\subsection{Decreasing Error Thresholds}
\label{subsec:apade}
When the number of samples $n$ increases, it is natural to have error threshold $\eta$ decrease. In this section, we consider the regime where the error threshold $\eta = \eta_n>0$ is a sequence decreasing to 0. In this setting, Theorem~\ref{thm:chernoff} still holds and the proof stays the same, except that we replace $\alpha_{\ell}$ and $\beta_{\ell}$, by the length-$n$ sequences $\alpha_{n, \ell}$ and $\beta_{n, \ell}$ for $\ell = 1, \ldots, n$, respectively, where  $\alpha_{n, \ell}$ and $\beta_{n, \ell}$ now depend on $\eta_n$ instead of a constant $\eta$: 
\begin{align}
\alpha_{n,1} &\delequal \frac{\sigma^2s^2 - 2\eta_n s}{2}, \label{alpha1n}\\ 
\alpha_{n,\ell} & = \frac{[a^2 + 2\sigma^2s(a+\eta_n)]\alpha_{n, \ell - 1} + \alpha_{n,1}}{1 - 2\sigma^2\alpha_{n, \ell-1}},\quad \forall \ell = 2,\ldots, n. \label{alphaElln}
\end{align}
The sequence $\beta_{n,\ell}$ is defined in a similar way. For Theorem~\ref{thm:upperLDP} to remain valid, we require $\eta_n$ no smaller than $1/\sqrt{n}$ to ensure that the right sides of~\eqref{eqn:chernoffUpper}-\eqref{eqn:chernoffLower} still converge to the right sides of~\eqref{eqn:PplusL}-\eqref{eqn:PminusL}, respectively. Let $\eta_n$ be a positive sequence such that 
\begin{align}
\eta_n = \omega\left(\frac{1}{\sqrt{n}}\right). \label{assumption:etan}
\end{align}
\begin{theorem}
\label{thm:decreasingeta}
For any $\sigma^2>0$ and $a>1$, let $\eta_n>0$ be a positive sequence satisfying~\eqref{assumption:etan}. Then, Theorem~\ref{thm:chernoff} holds with $\alpha_{\ell}$ replaced by $\alpha_{n, \ell}$, and $\beta_{\ell}$ by $\beta_{n, \ell}$, and Theorem~\ref{thm:upperLDP} holds with~\eqref{eqn:PplusL} and~\eqref{eqn:PminusL} replaced, respectively, by
\begin{align}
\liminf_{n\rightarrow \infty} P^+(n, a, \eta_n) &\geq \log a, \label{eqn:decreasingBDp}\\
\liminf_{n\rightarrow \infty} P^-(n, a, \eta_n) &\geq \log a.\label{eqn:decreasingBDm}
\end{align}
\end{theorem}

The proof of Theorem~\ref{thm:decreasingeta} is presented in Appendix~\ref{app:pfdecreasingeta} below. Theorem~\ref{thm:decreasingeta} is a quite strong result as it states that even if the error threshold is a sequence decreasing to zero, as long as~\eqref{assumption:etan} is satisfied, the probability of estimation error exceeding such decreasing thresholds is still exponentially small, with exponent being at least $\log a$. 

\begin{corollary}
\label{cor:disp}
For any $\sigma^2 > 0$ and any $a>1$, there exists a constant $c\geq \frac{1}{2}\log(a)$ such that for all $n$ large enough, 
\begin{align}
\prob{|\MLE- a| \geq \sqrt{\frac{\log\log n}{n}}} \leq 2e^{-cn}.
\end{align}
\end{corollary}

Corollary~\ref{cor:disp} is used in Section~\ref{subsec:dispersion} below to derive the dispersion of nonstationary Gauss-Markov sources. The proof of Corollary~\ref{cor:disp} is by applying~Theorem~\ref{thm:decreasingeta} with $\eta_n$ chosen as 
\begin{align}
\eta_n = \sqrt{\frac{\log\log n}{n}}.
\end{align}

\subsection{Generalization to sub-Gaussian $Z_i$'s}
\label{subsec:gen2sub}
In this section, we generalize the above results to the case where $Z_i$'s in~\eqref{eqn:GMmodel} are zero-mean sub-Gaussian random variables. This general result is of independent interest and will not be used in the rest of the paper. 
\begin{definition}[sub-Gaussian random variable, e.g. {\cite[Def. 2.7]{wainwright2019}}]
Fix $\sigma>0$. A random variable $Z\in\mathbb{R}$ with mean $\mu$ is said to be $\sigma$-sub-Gaussian with variance proxy $\sigma^2$ if its moment-generating function (MGF) satisfies 
\begin{align}
\mathbb{E}[e^{s(Z - \mu)}] \leq e^{\frac{\sigma^2 s^2}{2}}, \label{def:subG}
\end{align}
for all $s \in \mathbb{R}$.
\end{definition}

One important property of $\sigma$-sub-Gaussian random variables is the following well-known bound on the MGF of quadratic functions of $\sigma$-sub-Gaussian random variables.
\begin{lemma}[{\cite[Prop. 2]{rantzer2018concentration}}]
\label{lem:subG}
Let $Z$ be a $\sigma$-sub-Gaussian random variable with mean $\mu$. Then
\begin{align}
\mathbb{E} \left[ \exp (s Z^2) \right] \leq \frac{1}{\sqrt{1 - 2\sigma^2 s}}\exp\left(\frac{s \mu^2}{1 - 2\sigma^2 s}\right)\label{property:subG}
\end{align}
for any $s < \frac{1}{2\sigma^2}$.
\end{lemma}

Equality holds in~\eqref{def:subG} and~\eqref{property:subG} when $Z$ is Gaussian. In particular, the right side of~\eqref{property:subG} is the MGF of the noncentral $\chi^2$-distributed random variable $Z^2$.

\begin{theorem}[Generalization to sub-Gaussian case]
\label{thm:subgaussian}
Theorems~\ref{thm:chernoff}--\ref{thm:decreasingeta} and Lemma~\ref{lemma:limiting} remain valid for the estimator~\eqref{eqn:MLEintro} when $Z_i$'s in~\eqref{eqn:GMmodel} are i.i.d. zero-mean $\sigma$-sub-Gaussian random variables.
\end{theorem}

The generalizations of Theorems~\ref{thm:chernoff}--\ref{thm:decreasingeta} and Lemma~\ref{lemma:limiting} from Gaussian to sub-Gaussian $Z_i$'s only require minor changes in the corresponding proofs. See Appendix~\ref{app:subG} for the details.

\section{The Dispersion of a Nonstationary Gauss-Markov Source}
\label{subsec:preliminaries}

\subsection{Rate-distortion functions}
\label{subsubsec:RW}
For a generic random process $\{X_i\}_{i=1}^{\infty}$, the $n$-th order (informational) rate-distortion function $\mathbb{R}_{X_1^n}(d)$ is defined as
\begin{align}
\mathbb{R}_{X_1^n}(d) \delequal \inf_{\substack{P_{Y_1^n | X_1^n}:\\\EX{\dis{X_1^n}{Y_1^n}}\leq d}} ~\frac{1}{n}I(X_1^n; Y_1^n),
\label{eqn:nRDF}
\end{align}
where $X_1^n \triangleq (X_1, \ldots, X_n)'$ is the $n$-dimensional random vector determined by the random process, $I(X_1^n; Y_1^n)$ is the mutual information between $X_1^n$ and $Y_1^n$, $d$ is a given distortion threshold, and $\dis{\cdot}{\cdot}$ is the distortion measure defined in~\eqref{eqn:disdef} in Sec.~\ref{subsec:nrdt} above. The rate-distortion function $\mathbb{R}_X(d)$ is defined as
\begin{align}
\mathbb{R}_X(d) \delequal \limsup_{n\rightarrow \infty}~\mathbb{R}_{X_1^n}(d).
\end{align}
For a wide class of sources, $\mathbb{R}_X(d)$ has been shown to be equal to the minimum achievable source coding rate under the average distortion criterion, in the limit of $n\to\infty$, see~\cite{shannon1959coding} for discrete memoryless sources and~\cite{goblick1969coding} for general ergodic sources. In particular, Gray's coding theorem~\cite[Th. 2]{gray1970information} for the Gaussian autoregressive processes directly implies that for the Gauss-Markov source $\{U_i\}_{i=1}^{\infty}$ in~\eqref{eqn:GMmodel} for any $a\in\mathbb{R}$, its rate-distortion function $\rdf_{U}(d)$ equals the minimum achievable source coding rate under the average distortion criterion as $n$ tends to infinity. The $n$-th order rate-distortion function $\mathbb{R}_{U_1^n}(d)$ of the Gauss-Markov source is given by the $n$-th order reverse waterfilling, e.g.~\cite[Eq.~(22)]{gray1970information}: 
\begin{align}
\mathbb{R}_{U_1^n} (d) &= \frac{1}{n}\sum_{i = 1}^n \frac{1}{2}\log \max\left(\mu_{n, i},~\frac{\sigma^2}{\theta_n}\right), \label{eqn:NRDF1}\\
d &= \frac{1}{n}\sum_{i = 1}^n \min\left(\theta_n, \frac{\sigma^2}{\mu_{n,i}}\right),\label{eqn:NRDF2}
\end{align} 
where $\theta_n > 0$ is the $n$-th order water level, and $\mu_{n, i}$'s for $i\in [n]$ (sorted in nondecreasing order) are the eigenvalues of the $n\times n$ matrix $\mtx{F}'\mtx{F}$ with $\mtx{F}$ being an $n\times n$ lower triangular matrix defined as
\begin{align}
(\mtx{F})_{ij} \triangleq \begin{cases}
1, & i=j, \\
-a, & i = j+1,\\
0, & \text{otherwise.}
\end{cases}
\label{def:A}
\end{align}
One can check that $\sigma^2 (\mtx{F}'\mtx{F})^{-1}$ is the covariance matrix of $U_1^n$. The way that one uses~\eqref{eqn:NRDF1}-\eqref{eqn:NRDF2} is to first solve the $n$-th order water level $\theta_n$ using~\eqref{eqn:NRDF2} for a given distortion threshold $d$, and then to plug that water level into~\eqref{eqn:NRDF1} to obtain $\rdf_{U_1^n}(d)$. The rate-distortion function $\mathbb{R}_{U}(d)$ of the Gauss-Markov source is given by the limiting reverse waterfilling:
\begin{align}
\mathbb{R}_U(d) &= \frac{1}{2\pi}\int_{-\pi}^{\pi}\frac{1}{2}\log\max\LRB{g(w),~\frac{\sigma^2}{\theta}}~dw, \label{eqn:RWR}\\
d &= \frac{1}{2\pi}\int_{-\pi}^{\pi}\min\LRB{\theta,~\frac{\sigma^2}{g(w)}}~dw, \label{eqn:RWD}
\end{align}
where $\theta>0$ is the limiting water level and $g(w)$ is a function from $[-\pi, \pi]$ to $\mathbb{R}$ given by 
\begin{align}
g(w) &\delequal 1 + a^2 -2a\cos(w). \label{eqn:g}
\end{align}
The rate-distortion function of the Gaussian memoryless source $\{Z_i\}_{i=1}^{\infty}$ (the special case when $a$ is set to 0 in the Gauss-Markov model) is~\cite{shannon1959coding}
\begin{align}
\rdf_Z(d) = \max\left(0, ~\frac{1}{2}\log\frac{\sigma^2}{d}\right). \label{rdfZ}
\end{align}
One can obtain~\eqref{rdfZ} from~\eqref{eqn:RWR}-\eqref{eqn:RWD} by noting that $g(w) = 1$ for $a = 0$, which further simplifies~\eqref{eqn:RWD} to $d = \theta$, and~\eqref{eqn:RWR} to~\eqref{rdfZ}. See Fig.~\ref{fig:RD} for a plot of $\rdf_U(d)$ and $\rdf_Z(d)$.

\begin{figure}[!htb]
\centering
\includegraphics[width=0.48\textwidth]{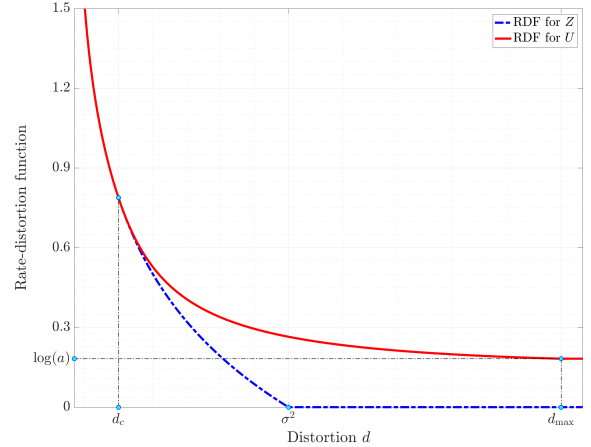}
\caption{Rate-distortion functions: $\mathbb{R}_{U}(d)$ in~\eqref{eqn:RWR} with $a = 1.2$, and $\mathbb{R}_{Z}(d)$ in~\eqref{rdfZ}.}
\label{fig:RD}
\end{figure}

\subsection{Operational Dispersion}
\label{subsubsec:dispersion}
To characterize the convergence rate of the minimum achievable source coding rate $R(n,d, \epsilon)$ (defined in~\eqref{def:Rnde} in Section~\ref{subsec:nrdt} above) to the rate-distortion function, we define the operational dispersion $V_U(d)$ for the Gauss-Markov source as
\begin{align}
V_U(d) \delequal \lim_{\epsilon\rightarrow 0}  \limsup_{n\rightarrow \infty} n\LRB{\frac{R(n,d, \epsilon) - \rdf_U(d)}{Q^{-1}(\epsilon)}}^2,
\label{eqn:opdis}
\end{align}
where $Q^{-1}$ denotes the inverse Q-function. The main result in the second part of this paper gives $V_U(d)$ for the nonstationary Gauss-Markov source.

\subsection{Informational Dispersion}
\label{subsubsec:dtiltedinfo}
The $\mathsf{d}$-tilted information~\cite[Def. 6]{kostina2012fixed} is the key random variable in our nonasymptotic analysis of $R(n, d, \epsilon)$. Under other names, the $\mathsf{d}$-tilted information has also been studied by Blahut~\cite[Th. 4]{blahut1972computation} and Kontoyiannis~\cite[Sec. III-A]{kontoyiannis2000pointwise}. Using the definition in~\cite[Def. 6]{kostina2012fixed}, the $\mathsf{d}$-tilted information $\jmath_{U_1^n}(u_1^n, d)$ in $u_1^n$ is 
\begin{align}
\jmath_{U_1^n}(u_1^n, d) \triangleq -\lambda_n^\star d - \log\mathbb{E}\exp\lpara{-\lambda_n^\star\mathsf{d}(u_1^n, V_1^{\star n})}, \label{dtiltedGM}
\end{align}
where $\lambda_n^\star$ is the negative slope of $\mathbb{R}_{U_1^n}(d)$ at the distortion level $d$ and $V_1^{\star n}$ is the random variable that achieves the infimum in~\eqref{eqn:nRDF} for $U_1^n$. In~\cite[Lem. 7, Eq. (228)]{dispersionJournal}, by a decorrelation argument, we obtained the following expression for the $\mathsf{d}$-tilted information for the Gauss-Markov source: for any $a\in\mathbb{R}$ and any $n\in\N$, 
\begin{align}
\jmath_{U_1^n}\left(u_1^n, d\right) & =  \sum_{i = 1}^n \frac{\min(\theta_n, \sigma_{n, i}^2)}{2\theta_n}\left(\frac{x_i^2}{\sigma_{n,i}^2} - 1\right)  + \notag \\&\quad \frac{1}{2}\sum_{i = 1}^n \log\frac{\max(\theta_n, \sigma_{n, i}^2)}{\theta_n}, \label{eqn:dtiexp}
\end{align}
where $\theta_n>0$ is given by~\eqref{eqn:NRDF2}, $x_1^n\triangleq \mtx{S}' u_1^n$ with $\mtx{S}$ being an $n\times n$ orthonormal matrix that diagonalizes $(\mtx{F}'\mtx{F})^{-1}$, and 
\begin{align}
\sigma_{n,i}^2\delequal \frac{\sigma^2}{\mu_{n,i}}
\label{eqn:sigmai}
\end{align} 
with $\mu_{n,i}$'s being the eigenvalues of the $n\times n$ matrix $\mtx{F}'\mtx{F}$. We refer to the random variable $X_1^n$,  defined by 
\begin{align}
X_1^n \triangleq \mtx{S}' U_1^n, \label{decor}
\end{align}
as the decorrelation of $U_1^n$. Note that the decorrelation $X_1^n$ has independent coordinates and 
\begin{align}
X_i \sim\mathcal{N}(0, \sigma_{n, i}^2).\label{eqn:Xi}
\end{align} 
Using~\eqref{eqn:NRDF1}-\eqref{eqn:NRDF2} and~\eqref{eqn:Xi}, one can show~\cite[Eq. (55) and (228)]{dispersionJournal} that the $\mathsf{d}$-tilted information $\jmath_{\bfU}(\bfu, d)$ in $\bfu$ for the Gauss-Markov source satisfies $\jmath_{\bfU}(\bfu, d) = \jmath_{\bfX}(\bfx, d)$. The minimum achievable source coding rates (defined in~\eqref{def:Rnde}) for lossy compression of $U_1^n$ and $X_1^n$ are equal, as are their rate-distortion functions: $\rdf_{\bfU}(d) = \rdf_{\bfX}(d)$, see~\cite[Sec. III.A]{dispersionJournal} for the details. It is known~\cite[Property 1]{kostina2012fixed} that the $\mathsf{d}$-tilted information $\jmath_{U_1^n}(u_1^n, d)$ satisfies (by the Karush-Kuhn-Tucker conditions for the optimization problem~\eqref{eqn:nRDF})
\begin{align}
\mathbb{E}\lbrac{\jmath_{U_1^n}(U_1^n, d)} = \mathbb{R}_{U_1^n}(d).
\end{align}
The informational dispersion $\mathbb{V}_U(d)$ is defined as the limit of the variance of the $\mathsf{d}$-tilted information normalized by $n$:
\begin{align}
\mathbb{V}_U(d) \delequal \limsup_{n\rightarrow \infty}~\frac{1}{n}\var{\jmath_{\bfU}(\bfU, d)}. \label{eqn:infdis}
\end{align}
By decorrelating the Gauss-Markov source $U_1^n$ and analyzing the limiting behavior of the eigenvalues of the covariance matrix of $U_1^n$, we obtain the following reverse waterfilling representation for the informational dispersion. The proof is given in Appendix~\ref{app:LemInfdis} below.
\begin{lemma}
\label{lemma:infdis}
The informational dispersion of the nonstationary Gauss-Markov source is given by 
\begin{align}
\mathbb{V}_U(d) = \frac{1}{4\pi}\int_{-\pi}^{\pi} \min\left[1,~\LRB{\frac{\sigma^2}{\theta g(w)}}^2\right]~dw,
\label{eqn:dispersion}
\end{align}
where $\theta>0$ is given in~\eqref{eqn:RWD}, and $g$ is in~\eqref{eqn:g}.
\end{lemma} 
Notice that the informational dispersion in the nonstationary case is given by the same expression as in the stationary case~\cite[Eq. (57)]{dispersionJournal}. It is known, e.g.~\cite[Eq. (94)]{kostina2012fixed} and~\cite[Sec. IV]{ingber2011dispersion}, that the informational dispersion for the Gaussian memoryless source $\{Z_i\}_{i=1}^{\infty}$ is 
\begin{align}
\mathbb{V}_Z(d) = \frac{1}{2},\quad\forall d\in (0, \sigma^2).
\label{disZ}
\end{align}
See Fig.~\ref{fig:DD} for a plot of $\mathbb{V}_U(d)$ and $\mathbb{V}_Z(d)$. 
\begin{figure}[!htb]
\centering
\includegraphics[width=0.48\textwidth]{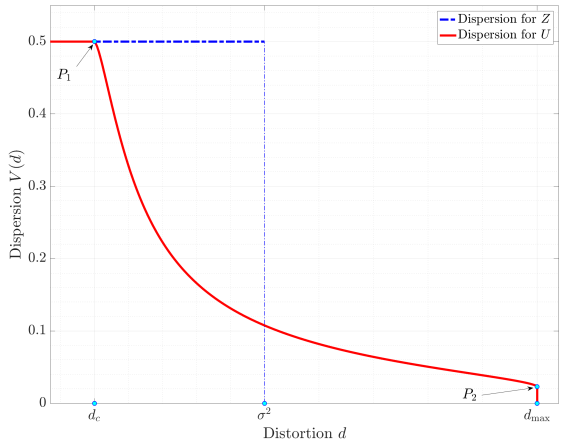}
\caption{Dispersions :$\mathbb{V}_{U}(d)$ in~\eqref{eqn:dispersion} with $a = 1.2$, and $\mathbb{V}_{Z}(d)$ in~\eqref{disZ}.}
\label{fig:DD}
\end{figure}

\subsection{A Few Remarks}
\label{subsubsec:CMD}
In view of~\eqref{eqn:RWD}, there are two special water levels $\theta_{\min}$ and $\theta_{\max}$, defined as follows:
\begin{align}
\theta_{\min} \triangleq \min_{w\in [-\pi, \pi]}~\frac{\sigma^2}{g(w)} = \frac{\sigma^2}{(a+1)^2}
\end{align}
and 
\begin{align}
\theta_{\max} \triangleq \max_{w\in [-\pi, \pi]}~\frac{\sigma^2}{g(w)} = \frac{\sigma^2}{(a-1)^2}.
\end{align}
The critical distortion $d_c$ is defined as the distortion corresponding to the water level $\theta_{\min}$. By~\eqref{eqn:RWD}, we have
\begin{align}
d_c = \theta_{\min} =  \frac{\sigma^2}{(a+1)^2}.
\end{align} 
The maximum distortion $d_{\max}$ is defined as the distortion corresponding to the water level $\theta_{\max}$. By~\eqref{eqn:RWD}, we have   
\begin{align}
d_{\max} = \frac{1}{2\pi}\int_{-\pi}^{\pi} \frac{\sigma^2}{g(w)}~dw.
\label{eqn:integraldmax}
\end{align} 
Using similar techniques as in~\cite[Eq. (169)--(172)]{dispersionJournal}, one can compute the integral in~\eqref{eqn:integraldmax} as
\begin{align}
d_{\max} = \frac{\sigma^2}{a^2 - 1}.
\end{align}
In this paper, we always consider a fixed distortion threshold $d$ such that $0 < d < d_{\max}$.

\begin{remark}
Gray~\cite[Eq. (24)]{gray1970information} showed the following relation between the rate-distortion function $\rdf_{U}(d)$ of the Gauss-Markov source and $\rdf_Z(d)$ of the Gaussian memoryless source: 
\begin{align}
\begin{cases}
\rdf_{U}(d) = \rdf_Z(d), & d\in (0, d_c], \\
\rdf_{U}(d) > \rdf_Z(d), & d\in ( d_c, \dmax).
\end{cases}
\label{rdf:comp}
\end{align}
Using Lemma~\ref{lemma:infdis} above, one can easily show (in the same way as~\cite[Cor. 1]{dispersionJournal}) that their dispersions are also comparable: 
\begin{align}
\begin{cases}
\mathbb{V}_{U}(d) = \mathbb{V}_Z(d), & d\in (0, d_c], \\
\mathbb{V}_{U}(d) < \mathbb{V}_Z(d), & d\in ( d_c, \sigma^2).
\end{cases}
\label{ddf:comp}
\end{align}
The results in~\eqref{rdf:comp}-\eqref{ddf:comp} imply that for low distortions $d\in (0, d_c)$, the minimum achievable source coding rate in compressing the Gauss-Markov source and the Gaussian memoryless source are the same up to second-order terms, a phenomenon we observed in the stationary case as well~\cite[Cor. 1]{dispersionJournal}. See Fig.~\ref{fig:RD} and Fig.~\ref{fig:DD} for a visualization of~\eqref{rdf:comp} and~\eqref{ddf:comp}, respectively.  
\end{remark}

\begin{remark}
For the function $\rdf_{U}(d)$, we show that
\begin{align}
\rdf_{U}(\dmax) = \log a. \label{rdfdmax}
\end{align}
This result has an interesting connection to the problem of control under communication constraints: in~\cite{wong1999systems} \cite[Th. 1]{baillieul1999feedback}~\cite[Prop. 3.1]{tatikonda2004control}, it was shown that the minimum rate to asymptotically stabilize a linear, discrete-time, scalar system is also $\log a$. The result in~\eqref{rdfdmax} implies that stability cannot be attained with any rate lower than $\log a$ even if an infinite lookahead is allowed. The derivation of~\eqref{rdfdmax} is presented in Appendix~\ref{app:rudmax} below.
\end{remark}

\begin{remark}
Let $P_1$ and $P_2$ be the two special points on the curve $\mathbb{V}_{U}(d)$ at distortions $d_c$ and $\dmax$, respectively. Then, the coordinates of $P_1$ and $P_2$ are given by  
\begin{align}
P_1 = (d_c, 1/2), \quad P_2 = \LRB{\dmax, \frac{(1+a^2)(a-1)}{2(a+1)^3}}. \label{p1p2}
\end{align}
The derivation for $P_2$ is the same as that in the stationary case~\cite[Eq. (61)]{dispersionJournal} except that we need to compute the residue at $1/a$ instead of at $a$ since we now have $a>1$, see~\cite[App. B-A]{dispersionJournal} for details.
\end{remark}

\subsection{Second-order Coding Theorem}
\label{subsec:dispersion}
Our main result establishes the equality between the operational dispersion and the informational dispersion. 

\begin{theorem}[Gaussian approximation]
\label{thm:dispersion}
For the Gauss-Markov source~\eqref{eqn:GMmodel} with $a>1$, any fixed excess-distortion probability $\epsilon\in (0,1)$, and distortion threshold $d\in(0,\dmax)$, it holds that 
\begin{align}
V_U(d)  = \mathbb{V}_U(d).  \label{eqn:opinf}
\end{align}
\end{theorem}
Specifically, we have the following converse and achievability.

\begin{theorem}[Converse]
\label{thm:converse}
For the Gauss-Markov source with $a>1$, any fixed excess-distortion probability $\epsilon\in (0,1)$, and distortion threshold $d$, the minimum achievable source coding rate  $R(n, d, \epsilon)$ satisfies 
\begin{align}
R(n, d, \epsilon) \geq \mathbb{R}_{U}(d) + \sqrt{\frac{\mathbb{V}_U(d)}{n}} Q^{-1}(\epsilon) -\frac{\log n}{2n} + O\left(\frac{1}{n}\right),
\end{align}
where $Q^{-1}$ denotes the inverse Q-function, $\mathbb{R}_U(d)$ is the rate-distortion function given in~\eqref{eqn:RWR}, and $\mathbb{V}_U(d)$ is the informational dispersion given by Lemma~\ref{lemma:infdis} above.
\end{theorem}

The converse proof is similar to that in the asymptotically stationary case in~\cite[Th. 7]{dispersionJournal}. See Appendix~\ref{app:pfConverse} for the details.

\begin{theorem}[Achievability]
\label{thm:achievability}
In the setting the Theorem~\ref{thm:converse}, the minimum achievable source coding rate $R(n, d, \epsilon)$ satisfies 
\begin{align}
R(n, d, \epsilon) \leq \mathbb{R}_{U} (d) + \sqrt{\frac{\mathbb{V}_U(d)}{n}} Q^{-1}(\epsilon) + O\left(\frac{1}{\sqrt{n}\log n}\right). \label{abound}
\end{align} 
\end{theorem}

Theorem~\ref{thm:dispersion} follows immediately from Theorems~\ref{thm:converse} and~\ref{thm:achievability}. Central to the achievability proof of Theorem~\ref{thm:achievability} is the following random coding bound: there exists an $(n, M, d, \epsilon)$ code such that~\cite[Cor. 11]{kostina2012fixed}
\begin{align}
\epsilon\leq \inf_{P_{V_1^n}}~\mathbb{E}\lbrac{\exp\lpara{-M\cdot P_{V_{1}^n}(\mathcal{B}(U_1^n, d))}},
\label{rcbound}
\end{align}
where the infimization is over all random variables defined on $\R^n$ and $\mathcal{B}(u_1^n, d))$ denotes the distortion $d$-ball around $u_1^n$: 
\begin{align}
\mathcal{B}(u_1^n, d)) \triangleq \lbpara{z_1^n\in\R^n\colon \mathsf{d}(u_1^n, z_1^n)\leq d}.
\end{align}
To obtain the achievability in~\eqref{abound} from~\eqref{rcbound}, we need to bound from below the probability $P_{V_{1}^n}(\mathcal{B}(U_1^n, d))$ that $V_1^n$ falls within the distortion $d$-ball $\mathcal{B}(U_1^n, d)$, where $V_1^n$ and $U_1^n$ are independent, in terms of the informational dispersion. This connection is made via the following second-order refinement of the ``lossy AEP" (asymptotic equipartition property~\cite[Lem. 1]{shannon1959coding}~\cite[Th. 1]{dembo2002source}~\cite[Lem. 2]{kostina2012fixed}) that applies to the nonstationary Gauss-Markov sources. 
\begin{lemma}[Second-order lossy AEP for the nonstationary Gauss-Markov sources]
\label{lemma:LossyAEP}
For the Gauss-Markov source with $a>1$, let $P_{V_1^{\star n}}$ be the random variable that attains the minimum in~\eqref{eqn:nRDF} with $X_1^n$ there replaced by $U_1^n$. It holds that
\begin{align}
 &\mathbb{P}\left[\log\frac{1}{P_{V_1^{\star n}}\left(\mathcal{B}\left(U_1^n, d\right)\right)} \geq \jmath_{U_1^n}(U_1^n, d) + p(n)\right] \leq  \frac{1}{q(n)},
\label{eqn:lossyAEP}
\end{align}
where 
\begin{align}
p(n) &\triangleq  c_1 (\log n)^{c_2} + c_3 \log n + c_4,\label{eqn:pn}\\
q(n) &\triangleq \Theta ( \log n ), \label{eqn:qn}
\end{align}
and $c_i$'s, $i = 1,...,4$, are positive constants depending only on $a$ and $d$.
\end{lemma}

The proof of Lemma~\ref{lemma:LossyAEP} is presented in Appendix~\ref{app:LossyAEP} below. The proof of Theorem~\ref{thm:achievability}, which uses uses the random coding bound~\eqref{rcbound} and Lemma~\ref{lemma:LossyAEP}, is presented in Appendix~\ref{app:pfAchievability} below.

\subsection{The Connection between Lossy AEP and Parameter Estimation}
\label{subsec:LossyAEPandPE}
The proof of lossy AEP in the form of Lemma~\ref{lemma:LossyAEP} is technical even for stationary memoryless sources~\cite[Lem. 2]{kostina2012fixed}. A lossy AEP for stationary $\alpha$-mixing processes was derived in~\cite[Cor. 17]{dembo2002source}. For stationary memoryless sources with single-letter distribution $P_{X}$, the idea in~\cite[Lem. 2]{kostina2012fixed} is to form a typical set $\mathcal{F}_n$ of source outcomes~\cite[Lem. 4]{kostina2012fixed} using the product of the empirical distributions~\cite[Eq. (270)]{kostina2012fixed}: $P_{\hat{X}} \times\ldots\times P_{\hat{X}}$, where $P_{\hat{X}}(x)\delequal \frac{1}{n}\sum_{i=1}^n\mathbbm{1}\{x_i = x\}$ is the empirical distribution of a given source sequence $x_1^n$, and then to show that the inequality inside the bracket in~\eqref{eqn:lossyAEP} holds for $x_1^n\in \mathcal{F}_n^c$ and that the probability of the complement set $\mathcal{F}_n^c$ is at most $1/q(n)$, where $p(n) = C\log n + c$ and $q(n) = K / \sqrt{n}$~\cite[Lem. 2]{kostina2012fixed}. The Gauss-Markov source is not memoryless, and it is nonstationary for $a>1$. To form a typical set of source outcomes, we define the following proxy random variables using the estimator $\mle$ in~\eqref{eqn:MLEintro}. 

\begin{definition}[Proxy random variables]
\label{def:proxy}
For each sequence $\bfu$ of length $n$ generated by the Gauss-Markov source, define the proxy random variable $\hat{X}_1^n$ as an $n$-dimensional Gaussian random vector with independent coordinates, each of which follows the distribution $\mathcal{N}(0, \hat{\sigma}_{n,i}^2)$ with 
\begin{align}
\hat{\sigma}_{n,1}^2 &\delequal \sigma^2\mle^{2n}, \label{eqn:sigma1hat}\\
\hat{\sigma}_{n,i}^2 &\delequal \frac{\sigma^2}{1 +\mle^2 - 2\mle\cos\frac{i\pi}{n+1} }, \quad 2\leq i\leq n,\label{eqn:sigmaihat}
\end{align}
where $\mle$ is in~\eqref{eqn:MLEintro} above.
\end{definition}

\begin{remark}
The proxy random variable in Definition~\ref{def:proxy} differs from that in~\cite[Eq. (119)]{dispersionJournal} for the stationary case in the behavior of the largest variance $\hat{\sigma}_{n,1}^2$. For each realization $u_1^n$, we construct the Gaussian random vector $\hat{X}_1^n$ according to~\eqref{eqn:sigma1hat}-\eqref{eqn:sigmaihat}, which is a proxy to the decorrelation $X_1^n$ in~\eqref{decor} above. The variances of $\hat{X}_i$ and $X_i$ are very close due to the closeness of $\mle$ to $a$ (Corollary~\ref{cor:disp}).
\end{remark}

\begin{remark}
Since the proxy random variable $\hat{X}_1^n$ depends on the realization of $\bfU$, Definition~\ref{def:proxy} defines the joint distribution of $(\bfX, \hat{X}_1^n)$, where $\bfX$ is the decorrelation of $\bfU$ in~\eqref{decor} above.
\end{remark}

The following convex optimization problem will be instrumental: for two generic random vectors $A_1^n$ and $B_1^n$ with distributions $P_{A_1^n}$ and $P_{B_1^n}$, respectively, define 
\begin{align}
\mathbb{R}(A_1^n, B_1^n, d) \delequal \inf_{\substack{P_{\bfF | A_1^n}:\\\EX{\dis{A_1^n}{\bfF}}\leq d}} ~\frac{1}{n}D(P_{\bfF |A_1^n} || P_{B_1^n}|P_{A_1^n}),
\label{eqn:crem}
\end{align}
where $D(P_{\bfF |A_1^n} || P_{B_1^n}|P_{A_1^n})$ is the conditional relative entropy. See~Appendix~\ref{app:paraCREM} for detailed discussions on this optimization problem.

For each realization $u_1^n$ (equivalently, each $x_1^n = \mtx{S}' u_1^n$ with the $n\times n$ matrix $\mtx{S}$ defined in the text above~\eqref{eqn:sigmai}), we define $n$ random variables $m_i(u_1^n)~, i = 1,\ldots,n$ as follows.  
\begin{itemize}
\item Let $X_1^n$ be the decorrelation of $\bfU$ in~\eqref{decor} above. Let $Y_1^{\star n}$ be the random variable that attains the infimum in $\R_{X_1^n}(d)$. 
\item For each $u_1^n$, choose $A_1^n$ in~\eqref{eqn:crem} to be the proxy random variable $\hat{X}_1^n$, and choose $B_1^n$ to be $Y_1^{\star n}$. Let $\hat{F}_1^{\star n}$ be the random variable that attains the infimum in $\R(\hat{X}_1^n, Y_1^{\star n}, d)$. 
\end{itemize}
Then, for each $i = 1, \ldots, n$, define 
\begin{align}
m_i(\bfu)\delequal \EX{(\hat{F}^\star_i- x_i)^2 | \hat{X}_i = x_i}.
\label{mi}
\end{align}
Denote 
\begin{align}
\eta_n\triangleq \sqrt{\frac{\log\log n}{n}}. \label{eqn:etan}
\end{align}
The typical set for the Gauss-Markov source is then defined as follows. 
\begin{definition}[Typical set]
\label{def:TS}
For any $d\in (0,\dmax)$, $n\geq 2$ and a constant $p>0$, define $\mathcal{T}(n,p)$ to be the set of vectors $\bfu \in \mathbb{R}^n$ that satisfy the following conditions: 
\begin{align}
\lrabs{\mle - a} &\leq \eta_n, \label{eqn:cond1} \\
\lrabs{\frac{1}{n}\sum_{i=1}^n \LRB{\frac{x_i^2}{\sigma_{n,i}^2}}^k - (2k-1)!!} &\leq 2,\quad k=1,2,3, \label{eqn:cond2}\\
\lrabs{\frac{1}{n}\sum_{i=1}^n m_i(\bfu) - d} &\leq p\eta_n,\label{eqn:cond3}
\end{align}
where $\bfx = \mtx{S}'\bfu$ is the decorrelation~\eqref{decor} and $\sigma_{n,i}^2$'s are defined in~\eqref{eqn:sigmai} above.
\end{definition}
The typical set in Definition~\ref{def:TS} is in the same form as that in the stationary case~\cite[Def. 2]{dispersionJournal}, but the definitions of proxy random variables and the analyses are different.
\begin{theorem}
\label{thm:typical}
For any $d\in (0,\dmax)$, there exists a constant $p>0$ such that the probability that the Gauss-Markov source produces a typical sequence satisfies
\begin{align}
\prob{\bfU \in \mathcal{T}(n,p) } \geq 1 - \Theta\LRB{\frac{1}{\log n}}.
\end{align}
\end{theorem}

Corollary~\ref{cor:disp} is essential to the proof of Theorem~\ref{thm:typical}. See the details in Appendix~\ref{app:PfThTS}.

Let $\mathcal{E}$ denote the event inside the square bracket in~\eqref{eqn:lossyAEP}. To prove Lemma~\ref{lemma:LossyAEP}, we intersect $\mathcal{E}$ with the typical set $\mathcal{T}(n,p)$ and the complement $\mathcal{T}(n,p)^{c}$, respectively, and then we bound the probability of the two intersections separately. See Appendix~\ref{app:LossyAEP} for the details.

\section{Discussion}
\label{sec:discussions}

\subsection{Stationary and Nonstationary Gauss-Markov Processes}
\label{subsubsec:diff}
It took several decades~\cite{kolmogorov1956shannon, berger1970information, gray1970information,hashimoto1980rate,gray2008note} to completely understand the difference in rate-distortion functions between stationary and nonstationary Gaussian autoregressive sources. We briefly summarize this subtle difference here to make the point that generalizing results from the stationary case to the nonstationary one is natural but nontrivial.

Since $\det(\mtx{F}) = 1$, the eigenvalues $\mu_{n,i}$'s of $\mtx{F}' \mtx{F}$ satisfy 
\begin{align}
\prod_{i = 1}^n \mu_{n,i} = 1. 
\label{eqn:prodMu}
\end{align}
Using~\eqref{eqn:prodMu}, we can equivalently rewrite~\eqref{eqn:NRDF1} as
\begin{align}
\rdf_{\bfU}(d) = \frac{1}{n}\sum_{i = 1}^n \max\LRB{0,~\frac{1}{2}\log\frac{\sigma_{n,i}^2}{\theta_n}},
\label{KNRDF}
\end{align}
where $\theta_n>0$ is in~\eqref{eqn:NRDF2} and $\sigma_{n,i}^2$'s are in~\eqref{eqn:sigmai}. Both~\eqref{eqn:NRDF1} and~\eqref{KNRDF} are valid expressions for the $n$-th order rate-distortion function $\rdf_{\bfU}(d)$, regardless of whether the source is stationary or nonstationary. The classical Kolmogorov reverse waterfilling result~\cite[Eq. (18)]{kolmogorov1956shannon}, obtained by taking the limit in~\eqref{KNRDF}, implies that the rate-distortion function of the \emph{stationary} Gauss-Markov source ($0<a<1$) is given by (the subscript K stands for Kolmogorov) 
\begin{align}
\rdf_{\text{K}}(d) = \frac{1}{2\pi}\int_{-\pi}^{\pi}\max\LRB{0,~\frac{1}{2}\log\frac{\sigma^2}{\theta g(w)}}~dw,
\label{Kol}
\end{align}
where $\theta>0$ is given in~\eqref{eqn:RWD} and $g(w)$ is given in~\eqref{eqn:g}. While~\eqref{eqn:RWR} and~\eqref{eqn:RWD} are valid for both stationary and nonstationary cases, Hashimoto and Arimoto~\cite{hashimoto1980rate} noticed in 1980 that~\eqref{Kol} is incorrect for the nonstationary Gaussian autoregressive source. The reason is the different asymptotic behaviors of the eigenvalues $\mu_{n,i}$'s of $\mtx{F}' \mtx{F}$~\eqref{def:A} in the stationary and nonstationary cases: while in the stationary case, the spectrum is bounded away from zero, in the nonstationary case, the smallest eigenvalue $\mu_{n,1}$ approaches 0, causing a discontinuity. By treating that smallest eigenvalue in a special way, Hashimoto and Arimoto~\cite[Th. 2]{hashimoto1980rate} showed that
\begin{align}
\rdf_{\text{HA}}(d) = \rdf_{\text{K}}(d) + \log(\max(a,1))\label{rdfHA}
\end{align}
is the correct rate-distortion function for both stationary and nonstationary Gauss-Markov sources, where the subscript HA stands for the authors of~\cite{{hashimoto1980rate}}. For the general higher-order Gaussian autoregressive source, the correction term needed in~\eqref{rdfHA} depends on the unstable roots of the characteristic polynomial of the source, see~\cite[Th. 2]{hashimoto1980rate} for the details. In 2008, Gray and Hashimoto~\cite{gray2008note} showed the equivalence between $\rdf_{\text{HA}}(d)$ in~\eqref{rdfHA}, obtained by taking a limit in~\eqref{KNRDF}, and Gray's result~$\rdf_{U}(d)$ in~\eqref{eqn:RWR}, obtained by taking a limit in~\eqref{eqn:NRDF1}. 

The tool that allows one to take limits in~\eqref{KNRDF} and~\eqref{eqn:NRDF1} is the following theorem on the asymptotic eigenvalue distribution of the almost Toeplitz matrix $\mtx{F}' \mtx{F}$, which is the (rescaled) inverse of the covariance matrix of $U_1^n$. Denote 
\begin{align}
\alpha\delequal \min_{w\in [-\pi, \pi]}~g(w) = (a-1)^2,
\label{ma}
\end{align}
and 
\begin{align}
\beta\delequal \max_{w\in [-\pi, \pi]}~g(w) = (a+1)^2.
\label{Mb}
\end{align}
Gray~\cite[Th. 2.4]{gray2006toeplitz} generalized the result of Grenander and Szeg{\"o}~\cite[Th. in Sec. 5.2]{grenander1984toeplitz} on the asymptotic eigenvalue distribution of Toeplitz forms to that of matrices that are asymptotically equivalent to Toeplitz forms, see~\cite[Chap. 2.3]{gray2006toeplitz} for the details. Define
\begin{align}
\alpha'\delequal \inf_{n\in\mathbb{N},~i\in[n]}~\mu_{n,i}.
\label{mprime}
\end{align}
\begin{theorem}[Gray~{\cite[Eq. (19)]{gray1970information}}, Hashimoto and Arimoto~{\cite[Th. 1]{hashimoto1980rate}}]
For any continuous function $F(t)$ over the interval 
\begin{align}
t\in\left[\alpha',~\beta\right],
\label{interval_t}
\end{align}
the eigenvalues $\mu_{n,i}$'s of $\mtx{F}' \mtx{F}$ with $\mtx{F}$ in~\eqref{def:A} satisfy 
\begin{align}
\lim_{n\rightarrow\infty}\frac{1}{n}\sum_{i = 1}^n F(\mu_{n,i}) = \frac{1}{2\pi}\int_{-\pi}^{\pi} F\left(g(w)\right)~dw,
\label{eqn:limiting_equality}
\end{align}
where $g(w)$ is defined in~\eqref{eqn:g}.
\label{thm:LimitingThm}
\end{theorem}

The eigenvalues $\mu_{n,i}$'s behave quite differently in the following three cases, leading to the subtle difference in the corresponding rate-distortion functions. 
\begin{enumerate}
\item For the stationary case $a\in (0,1)$, it can be easily shown~\cite[Eq. (71)]{dispersionJournal} that $\alpha' = \alpha > 0$ and all eigenvalues $\mu_{n,i}$'s lie in between $\alpha$ and $\beta$. Kolmogorov's formula~\eqref{Kol} is obtained by applying Theorem~\ref{thm:LimitingThm} to~\eqref{KNRDF} using the function 
\begin{align}
F_{\text{K}}(t) \delequal \max\left(0,~\frac{1}{2}\log\frac{\sigma^2}{\theta t}\right),
\label{FK}
\end{align}
where $\theta>0$ is given by~\eqref{eqn:RWD}.

\item For the Wiener process ($a = 1$), closed-form expressions of $\mu_{n,i}$'s are given by Berger~\cite[Eq. (2)]{berger1970information}. Those results imply that the smallest eigenvalue $\mu_{n,1}$ is of order $\Theta\LRB{\frac{1}{n^2}}$, and thus $\alpha' = \alpha = 0$. Using the same function as in~\eqref{FK}, Berger obtained the rate-distortion functions for the Wiener process~\cite[Eq. 4]{berger1970information}~\footnote{To be precise, although the rate-distortion function for the Wiener process is correct in~\cite[Eq. 4]{berger1970information}, the proof there is not rigorous since in this case $\alpha' = \alpha = 0$ but $F_{\text{K}}(t)$ is not continuous at $t = 0$ as pointed out in~\cite[Eq. (23)]{gray2008note}. Therefore, the limit leading to~\cite[Eq. 4]{berger1970information} needs extra justifications.}.

\item For the nonstationary case $a >1$, we have $\alpha' = 0 < \alpha$, the smallest eigenvalue $\mu_{n,1}$ is of order $\Theta(a^{-2n})$ and the other $n-1$ eigenvalues lie in between $\alpha$ and $\beta$. This behavior of eigenvalues was shown by Hashimoto and Arimoto~\cite[Lemma]{hashimoto1980rate} for higher-order Gaussian autoregressive sources, and we will show a refined version for the Gauss-Markov source in Lemma~\ref{lemma:eigScaling} below. As pointed out in~\cite[Th. 1]{hashimoto1980rate}, an application of Theorem~\ref{thm:LimitingThm} using the function~\eqref{FK} fails to yield the correct rate-distortion function for nonstationary sources due to the discontinuity of $F_{\text{K}}(t)$ at 0. Gray~\cite[Eq. (22)]{gray1970information} and Hashimoto and Arimoto~\cite{hashimoto1980rate} circumvent this difficulty in two different ways, which lead to~\eqref{eqn:RWR} and~\eqref{rdfHA}, respectively. Gray~\cite[]{gray1970information} applied Theorem~\ref{thm:LimitingThm} on~\eqref{eqn:NRDF1} using the function
\begin{align}
F_{\text{G}}(t) = \frac{1}{2}\log\max\left(t,~\frac{\sigma^2}{\theta}\right),
\label{func:gray}
\end{align}
which is indeed continuous at $0$, while Hashimoto and Arimoto~\cite[Th. 2]{hashimoto1980rate} still use the function $F_{\text{K}}(t)$ but consider $\mu_{n, 1}$ and $\mu_{n,i},~i\geq 2$ separately: 
\begin{align}
\frac{1}{n}\sum_{i = 2}^n F_{\text{K}}(\mu_{n,i}) + \frac{1}{n} F_{\text{K}}(\mu_{n,1}),
\end{align}
which in the limit yields~\eqref{rdfHA} by plugging $\mu_{n,1} = \Theta(a^{-2n})$ into~\eqref{FK}.
\end{enumerate}

\subsection{New Results on the Spectrum of the Covariance Matrix}
\label{subsec:spectrum}

The following result on the scaling of the eigenvalues $\mu_{n,i}$'s refines~\cite[Lemma]{hashimoto1980rate}. Its proof is presented in Appendix~\ref{app:proof_eigScaling}.
\begin{lemma}
\label{lemma:eigScaling}
Fix $a>1$. For any $i = 2, \ldots, n$, the eigenvalues of $\mtx{F}' \mtx{F}$~\eqref{def:A} are bounded as 
\begin{align}
\xi_{n-1, i - 1} \leq \mu_{n,i} \leq \xi_{n, i}, \label{eqn:eig2N}
\end{align}
where
\begin{align}
\xi_{n, i} \delequal 1 + a^2 - 2a \cos\left(\frac{i\pi}{n +1}\right). \label{eqn:xini}
\end{align}
The smallest eigenvalue is bounded as
\begin{align}
2\log a +\frac{c_2}{n} \geq -\frac{1}{n}\log \mu_{n,1} \geq  2\log a-\frac{c_1}{n}, \label{eqn:eig1}
\end{align}
where $c_1>0$ and $c_2$ are constants given by 
\begin{align}
c_1 &= 2\log (a+1) + \frac{a\pi}{a^2 -1}, \label{int:c1}  \\
c_2 &= 2\log\frac{a}{a^2 -1} + \frac{2a\pi}{a^2  - 1}. \label{int:c2} 
\end{align}
\end{lemma}

\begin{remark}
The constant $c_1$ in~\eqref{int:c1} is positive, while $c_2$ in~\eqref{int:c2} can be positive, zero or negative, depending on the value of $a>1$. Lemma~\ref{lemma:eigScaling} indicates that $a^{-2n}$ is a good approximation to $\mu_{n,1}$. Using~\eqref{eqn:eig2N}--\eqref{eqn:xini}, we deduce that for $i= 2, \ldots, n$, 
\begin{align}
\mu_{n,i}  \in [\alpha, \beta].
\end{align}
\end{remark}

Based on Lemma~\ref{lemma:eigScaling}, we obtain a nonasymptotic version of Theorem~\ref{thm:LimitingThm}, which is useful in the analysis of the dispersion, in particular, in deriving Proposition~\ref{prop:approx} in Appendix~\ref{app:EV} below.

\begin{theorem}
\label{thm:nonasymEig}
Fix any $a>1$. For any bounded, $L$-Lipschitz and nondecreasing function (or nonincreasing function) $F(t)$ over the interval~\eqref{interval_t} and any $n\geq 1$, the eigenvalues $\mu_{n,i}$'s of $\mtx{F}' \mtx{F}$~\eqref{def:A} satisfy 
\begin{align}
\lrabs{ \frac{1}{n}\sum_{i = 1}^n F(\mu_{n,i}) - \frac{1}{2\pi}\int_{-\pi}^{\pi} F\left(g(w)\right)~dw} \leq \frac{C_L}{n},
\label{eqn:nonasymEig}
\end{align}
where $g(w)$ is defined in~\eqref{eqn:g} and $C_L > 0$ is a constant that depends on $L$ and the maximum absolute value of $F$.
\end{theorem}

The proof of Theorem~\ref{thm:nonasymEig} is in Appendix~\ref{app:nonasymEig}.

\section{Conclusion}
\label{sec:con}
In this paper, we obtain nonasymptotic (Theorem~\ref{thm:chernoff}) and asymptotic (Theorem~\ref{thm:upperLDP}) bounds on the estimation error of the maximum likelihood estimator of the parameter $a$ of the nonstationary scalar Gauss-Markov process. Numerical simulations in Fig.~\ref{fig:compare} confirm the tightness of our estimation error bounds compared to previous works. As an application of the estimation error bound (Corollary~\ref{cor:disp}), we find the dispersion for lossy compression of the nonstationary Gauss-Markov sources (Theorems~\ref{thm:converse} and~\ref{thm:achievability}). Future research directions include generalizing the error exponent bounds in this paper, applicable to identification of scalar dynamical systems, to vector systems, and finding the dispersion of the Wiener process.
\appendices

\section{}
\label{app:pf}

\subsection{Proof of Theorem~\ref{thm:chernoff}}
\label{app:pfThmChernoff}
\begin{proof}
We present the proof of~\eqref{eqn:chernoffUpper}. The proof of~\eqref{eqn:chernoffLower} is similar and is omitted. For any $n\geq 2$, denote by $\mathcal{F}_n$ the $\sigma$-algebra generated by $Z_1, \ldots, Z_n$. For any $s>0$, $\eta>0$, and $n\geq 2$, we denote the following random variable
\begin{align}
W_n \delequal \exp\left ( s\sum_{i= 1}^{n-1} \left (U_i Z_{i+1} - \eta U_i^2\right )\right ).
\end{align}
By the Chernoff bound, we have 
\begin{align}
\prob{\MLE - a \geq \eta} \leq \inf_{s>0}~\mathbb{E}\left[W_n\right].
\end{align}
To compute $\mathbb{E}\left[W_n\right]$, we first consider the conditional expectation $\mathbb{E}\left[W_n|\mathcal{F}_{n-1}\right]$. Since $Z_n$ is the only term in $W_n$ that does not belong to $\mathcal{F}_{n-1}$, we have 
\begin{align}
& \mathbb{E} \left[W_n\right] \notag\\
= & \mathbb{E}\left [ W_{n-1}\cdot\mathbb{E} \left [ \exp\left (s\left (U_{n-1}Z_n-\eta U_{n-1}^2\right )\right )|\mathcal{F}_{n-1} \right ]\right ] \label{steps:0}\\
= & \EX{ W_{n-1}\cdot \exp\left (\alpha_1 U_{n-1}^2 \right ) },\label{steps:1}
\end{align}
where $\alpha_1$ is the deterministic function of $s$ and $\eta$ defined in~\eqref{alpha1}, and~\eqref{steps:1} follows from the moment generating function of $Z_n$. To obtain a recursion, we then consider the conditional expectation $\mathbb{E}\left[W_{n-1}\cdot \exp\left (\alpha_1 U_{n-1}^2 \right ) | \mathcal{F}_{n-2} \right]$. Since $U_{n-1}^2$ and $U_{n-2}Z_{n-1}$ are the only two terms in $W_{n-1}\cdot \exp(\alpha_1 U_{n-1}^2)$ that do not belong to $\mathcal{F}_{n-2}$, we use the relation  $U_{n-1} = aU_{n-2} + Z_{n-1}$ and we complete squares in $Z_{n-1}$ to obtain 
\begin{align}
& W_{n-1}\cdot \exp\left (\alpha_1 U_{n-1}^2\right )\notag \\
= & W_{n-2}\cdot\exp\Bigg ( \alpha_1\left(Z_{n-1}+\left (a+ \frac{s}{ 2\alpha_1} \right )U_{n-2}\right)^2 + \notag \\
&\quad \left (a^2\alpha_1 - s\eta\right )U_{n-2}^2 - \alpha_1\left(a+\frac{s}{2\alpha_1}\right)^2 U_{n-2}^2\Bigg).
\end{align}
Furthermore, using the formula for the moment generating function of the noncentral $\chi^2$-distributed random variable 
\begin{align}
\left(Z_{n-1}+\left(a+ \frac{s}{ 2\alpha_1}\right)U_{n-2}\right)^2
\end{align}
with 1 degree of freedom, we obtain 
\begin{align}
& \EX{ W_{n-1}\cdot \exp\left (\alpha_1 U_{n-1}^2\right ) } \notag \\
= & \frac{1}{\sqrt{1-2\sigma^2 \alpha_1}}\EX{W_{n-2}\cdot \exp\left (\alpha_2 U_{n-2}^2\right)}.\label{diff}
\end{align}
This is where our method diverges from Rantzer~\cite[Lem. 5]{rantzer2018concentration}, who chooses $s = \frac{\eta}{\sigma^2}$ and bounds $\alpha_2\leq \alpha_1$ (due to Property A4 in Appendix~\ref{app:seqA} below) in~\eqref{diff}. Instead, by conditioning on $\mathcal{F}_{n-3}$ in~\eqref{diff} and repeating the above recursion for another $n-2$ times, we compute $\EX{W_{n}}$ exactly using the sequence $\{\alpha_\ell\}$:
\begin{align}
\EX{W_{n}} = \exp\left ( -\frac{1}{2}\sum_{\ell = 1}^{n-1}\log\left (1-2\sigma^2\alpha_{\ell} \right )\right ).
\end{align}
If $s\not\in \mathcal{S}_n^+$, then by the definition of the set $\mathcal{S}_n^+$ we have $\EX{W_n} = +\infty$. Therefore, 
\begin{align}
\inf_{s>0}~\mathbb{E}\left[W_n\right] = \inf_{s\in \mathcal{S}_n^+}~\mathbb{E}\left[W_n\right].
\end{align}
\end{proof}

\subsection{Properties of the Sequence $\alpha_{\ell}$}
\label{app:seqA}
We derive several important elementary properties of the sequences $\alpha_{\ell}$ and $\beta_{\ell}$. First, we consider $\alpha_{\ell}$. We find the two fixed points $r_1 < r_2$ of the recursive relation~\eqref{alphaEll} by solving the following quadratic equation in $x$: 
\begin{align}
2\sigma^2 x^2 + [a^2 + 2\sigma^2 s(a + \eta) - 1]x + \alpha_1 = 0. \label{eqn:sol1}
\end{align}

\subsubsection*{Property A1} \label{sec:pA1}
For any $s>0$ and $\eta>0$,~\eqref{eqn:sol1} has two roots $r_1< r_2$, and $r_1<0$. The two roots $r_1$ and $r_2$ are given by 
\begin{align}
r_1 &= \frac{-[a^2 + 2\sigma^2 (a+\eta)s - 1] - \sqrt{\Delta}}{4\sigma^2},\label{root1} \\
r_2 &= \frac{-[a^2 + 2\sigma^2 (a+\eta)s - 1] + \sqrt{\Delta}}{4\sigma^2},\label{root2}
\end{align}
where $\Delta$ denotes the discriminant of~\eqref{eqn:sol1}:
\begin{align}
\Delta &= 4\sigma^4 [(a+\eta)^2 - 1] s^2 + \notag \\ 
&\quad 4\sigma^2 [(a+\eta)(a^2 -1) + 2\eta] s+ (a^2 - 1)^2.\label{discriminantalpha}
\end{align}
\begin{proof}
Note that the discriminant $\Delta$ satisfies 
\begin{align}
\Delta > (a^2 - 1)^2 > 0, \label{positiveD}
\end{align}
where we used $a>1$. Then,~\eqref{root1} implies $r_1 < 0$.
\end{proof}

\subsubsection*{Property A2} \label{sec:pA2}
For $ \frac{2\eta}{\sigma^2} \neq s>0$ and $\eta>0$, the sequence $\frac{\alpha_\ell - r_1}{\alpha_\ell - r_2}$ is a geometric sequence with common ratio 
\begin{align}
q \delequal \frac{[a^2 + 2\sigma^2 s(a+\eta)] + 2\sigma^2 r_1}{[a^2 + 2\sigma^2 s(a+\eta)] + 2\sigma^2 r_2}. \label{cr}
\end{align}
Furthermore, 
\begin{align}
q \in (0,1), \label{crrange}
\end{align}
and it follows immediately that 
\begin{align}
\alpha_{\ell} & = r_1 + \frac{(r_1 - r_2)\frac{\alpha_1 - r_1}{\alpha_1-r_2}q^{\ell-1}}{1 - \frac{\alpha_1 - r_1}{\alpha_1-r_2}q^{\ell-1}}, \label{eqn:expAlpha} \\
&= r_2 + \frac{r_2 - r_1}{ \frac{\alpha_1 - r_1}{\alpha_1-r_2}q^{\ell-1} - 1}. \label{eqn:expAlpha2}
\end{align}
\begin{proof}
Using the recursion~\eqref{alphaEll} and the fact that $r_1$ and $r_2$ are the fixed points of~\eqref{alphaEll}, one can verify that $\frac{\alpha_\ell - r_1}{\alpha_\ell - r_2}$ is a geometric sequence with common ratio $q$ given by~\eqref{cr}. The relation~\eqref{crrange} is verified by direct computations using~\eqref{root1} and~\eqref{root2}. 
\end{proof}

\subsubsection*{Property A3} \label{sec:pA3}
For any $\frac{2\eta}{\sigma^2} \neq s>0$ and $\eta>0$, we have 
\begin{align}
\lim_{\ell\rightarrow \infty}\alpha_{\ell} = r_1. \label{seqLim}
\end{align}
For $s = \frac{2\eta}{\sigma^2}$, we have $\alpha_{\ell} = 0 = r_2 > r_1,~\forall\ell \geq 1$.

\begin{proof}
The limit~\eqref{seqLim} follows from~\eqref{crrange} and~\eqref{eqn:expAlpha}. Plugging $s = \frac{2\eta}{\sigma^2}$ into~\eqref{alpha1} yields $\alpha_1 = 0$, which implies by~\eqref{alphaEll} that $\alpha_{\ell} = 0$ for $\ell\geq 1$. 
\end{proof}

\subsubsection*{Property A4} \label{sec:pA4}
For any $0 < s\leq \frac{2\eta}{\sigma^2}$, we have $\alpha_\ell < 0$ and $\alpha_{\ell}$ decreases to $r_1$ geometrically. For $s > \frac{2\eta}{\sigma^2}$,~\eqref{seqLim} still holds, but the convergence is not monotone: there exists an $\ell^\star\geq 1$ such that  $\alpha_\ell > 0$ and increases to $\alpha_{\ell^\star}$ for $1\leq \ell \leq \ell^\star$; and $\alpha_\ell < 0$ and increases to $r_1$ for $\ell > \ell^\star$.

\begin{proof}
Due to~\eqref{eqn:expAlpha2}, the monotonicity of $\alpha_{\ell}$ depends on the signs of $r_2 - r_1$ and $\frac{\alpha_1 - r_1}{\alpha_1 - r_2}$. Note that $r_2 - r_1 > 0$ by Property A1. Plugging $x = \alpha_1$ into~\eqref{eqn:sol1}, we have 
\begin{align}
(\alpha_1 - r_1)(\alpha_1 - r_2) = (a+\sigma^2 s)^2 \alpha_1. \label{pr}
\end{align}
Since for $0 < s\leq \frac{2\eta}{\sigma^2}$, we have $\alpha_1 < 0$ by~\eqref{alpha1}; we must also have $\frac{\alpha_1 - r_1}{\alpha_1 - r_2} < 0$ by~\eqref{pr}. Due to~\eqref{eqn:expAlpha} and~\eqref{eqn:expAlpha2}, this immediately implies that $\alpha_{\ell}$ decreases to $r_1$. Therefore, $\alpha_{\ell}\leq \alpha_1 <0,~\forall\ell\geq 1$. For any $s > \frac{2\eta}{\sigma^2}$, we have $\alpha_1 > 0$ and $\frac{\alpha_1 - r_1}{\alpha_1 - r_2} > 0$. In fact, since $r_1 < 0$, we have $\alpha_1 > r_2$, which implies  $\frac{\alpha_1 - r_1}{\alpha_1 - r_2} > 1$. Therefore, the conclusion follows from~\eqref{eqn:expAlpha2}.
\end{proof}

\subsubsection*{Property A5} \label{sec:pA5}
For any $\eta > 0$, the root $r_1$ in~\eqref{root1} is a decreasing function in $s > 0$.

\begin{proof}
Direct computations using~\eqref{root1},~\eqref{discriminantalpha} and the assumption that  $a>1$.
\end{proof}

\subsection{Properties of the Sequence $\beta_{\ell}$}
\label{app:seqB}
The sequence $\beta_{\ell}$  is analyzed similarly, although it is slightly more involved than $\alpha_{\ell}$. We only consider $0<s\leq \frac{2\eta}{\sigma^2}$ in the rest of this section. We find the two fixed points $t_1 < t_2$ of the recursive relation~\eqref{betaEll} by solving the following quadratic equation in $x$: 
\begin{align}
2\sigma^2 x^2 + [a^2 + 2\sigma^2 s(-a + \eta) - 1]x + \beta_1 = 0. \label{eqn:sol2}
\end{align} 

\subsubsection*{Property B1} \label{sec:pB1}
For $s = \frac{2\eta}{\sigma^2}$, we have $\beta_{\ell} = 0,~\forall\ell\geq 1$. For any $\eta > 0$ and $0 < s\leq \frac{2\eta}{\sigma^2}$,~\eqref{eqn:sol2} has two distinct roots $t_1 < 0 < t_2$, given by 
\begin{align}
t_1 &= \frac{-[a^2+2\sigma^2 s(-a +\eta) - 1] - \sqrt{\Gamma}}{4\sigma^2}, \label{t1} \\
t_2 &= \frac{-[a^2+2\sigma^2 s(-a +\eta) - 1] + \sqrt{\Gamma}}{4\sigma^2} \label{t2},
\end{align}
where the discriminant $\Gamma$ of~\eqref{eqn:sol2} is 
\begin{align}
\Gamma &= 4\sigma^4 [(-a+\eta)^2 - 1] s^2 + \notag \\ 
&\quad 4\sigma^2 [(-a+\eta)(a^2 -1) + 2\eta] s+ (a^2 - 1)^2. \label{discriminantbeta}
\end{align}

\begin{proof}
We verify that $\Gamma > 0$ for any $\eta > 0$ and $0 < s\leq \frac{2\eta}{\sigma^2}$.  The reason that $\Gamma > 0$ is not as obvious as~\eqref{positiveD} is due to the subtle difference between~\eqref{discriminantalpha} and~\eqref{discriminantbeta} in the negative sign of $a$. Note that $\Gamma$ in~\eqref{discriminantbeta} is a quadratic equation in $s$ and the discriminant of~$\Gamma$ is given by 
\begin{align}
\gamma = 16\sigma^4 (2a\eta - a^2 + 1)^2 \geq 0.
\end{align}
Hence, in general,~\eqref{discriminantbeta} has two roots (distinct when $\eta\neq \frac{a^2 - 1}{2a}$) and $\Gamma$ could be positive or negative. However, an analysis of two cases $(-a+\eta)^2 -1 \geq 0$ and $(-a+\eta)^2 -1 <0$ reveals that $\Gamma > 0$ for any $\eta > 0$ and $0 < s\leq \frac{2\eta}{\sigma^2}$. Therefore,~\eqref{eqn:sol2} has two distinct roots $t_1 < t_2$ given in~\eqref{t1} and~\eqref{t2} above. From~\eqref{eqn:sol2}, we have $t_1t_2 = \frac{\beta_1}{2\sigma^2}$, which is negative for $0 < s\leq \frac{2\eta}{\sigma^2}$. Therefore, we have $t_1 < 0 < t_2$.
\end{proof}

\subsubsection*{Property B2} \label{sec:pB2}
For any $\eta > 0$ and $0 < s\leq \frac{2\eta}{\sigma^2}$, the sequence $\frac{\beta_{\ell} - t_1}{\beta_{\ell} - t_2}$ is a geometric sequence with common ratio
\begin{align}
p  \delequal \frac{[a^2 + 2\sigma^2 s(-a+\eta)] + 2\sigma^2 t_1}{[a^2 + 2\sigma^2 s(-a+\eta)] + 2\sigma^2 t_2}. \label{crbeta}
\end{align}
In addition, for any $\eta > 0$ and $0 < s\leq \frac{2\eta}{\sigma^2}$, we also have  
\begin{align}
p\in (0,1). \label{crbetarange}
\end{align}
It follows immediately that 
\begin{align}
\beta_{\ell} & = t_1 + \frac{(t_1 - t_2)\frac{\beta_1 - t_1}{\beta_1-t_2}p^{\ell-1}}{1 - \frac{\beta_1 - t_1}{\beta_1-t_2}p^{\ell-1}}, \label{eqn:expBeta} \\
&= t_2 + \frac{t_2 - t_1}{ \frac{\beta_1 - t_1}{\beta_1-t_2}p^{\ell-1} - 1}. \label{eqn:expBeta2}
\end{align}

\begin{proof}
Similar to that of Property A2 above for $\alpha_\ell$.
\end{proof}

\subsubsection*{Property B3} \label{sec:pB3}
For any $\eta > 0$ and $0 < s\leq \frac{2\eta}{\sigma^2}$, we have $\beta_{\ell} \leq \beta_1 < 0$, and $\beta_{\ell}$ decreases to $t_1$ geometrically:
\begin{align}
\lim_{\ell\rightarrow \infty} \beta_{\ell} = t_1.
\end{align}

\begin{proof}
This can be verified using~\eqref{eqn:expBeta} and~\eqref{eqn:expBeta2} by noticing that $t_2 - t_1 > 0$ and that for $0 < s\leq \frac{2\eta}{\sigma^2}$,
\begin{align}
(\beta_1 - t_1)(\beta_2 - t_2) = (a-\sigma^2 s)^2 \beta_1 < 0.
\end{align}
\end{proof}

\subsubsection*{Property B4} \label{sec:pB4}
For any constant $a>1$, the two thresholds $\eta_1$ and $\eta_2$, defined in~\eqref{eta1} and~\eqref{eta2}, respectively, satisfy the following 
Then, 
\begin{enumerate}
\item When $0< \eta \leq \eta_1$, the root $t_1$ in~\eqref{t1} is an increasing function in $s \in \mathcal{I}_\eta$. 

\item When $\eta \geq \eta_2$, $t_1$ is a decreasing function in $s \in \mathcal{I}_\eta$. 

\item When $\eta_1 < \eta < \eta_2$, $t_1$ is a decreasing function in $s \in (0, s^\star)$ and an increasing function in $s\in \left(s^\star, \frac{2\eta}{\sigma^2}\right)$, where $s^\star$ is the unique solution in the interval $\mathcal{I}_\eta$ to 
\begin{align}
\frac{d t_1}{ds} \Big |_{s = s^\star} = 0,\label{eqn:dt1}
\end{align}
and $s^\star$ is given by 
\begin{align}
s^\star \triangleq \frac{a\eta (\eta - \eta_1)}{\sigma^2 (1 - (\eta -a)^2)}.\label{def:sstar}
\end{align}
\end{enumerate}

\begin{proof}
Using~\eqref{t1} and~\eqref{discriminantbeta}, we compute the derivatives of $t_1$ as follows:
\begin{align}
\frac{d t_1}{ds} &= -\frac{\eta - a}{2} - \frac{1}{\sqrt{\Gamma}}\Bigg\{ \sigma^2 [(-a+\eta)^2 - 1]s \notag \\ & \quad\quad + \frac{1}{2} [(-a+\eta)(a^2 - 1) + 2\eta]\Bigg\}, \label{eqn:firstderivative}\\
\frac{d^2 t_1}{ds^2} &= \frac{\sigma^2(2a\eta - a^2 + 1)^2}{\Gamma^{\frac{3}{2}}}\geq 0.\label{eqn:doublederivative}
\end{align}
To simplify notations, denote by $L(s)$ the first derivative: 
\begin{align}
L(s)\triangleq \frac{dt_1}{ds}(s).
\end{align}
From~\eqref{eqn:firstderivative}, we have
\begin{align}
L(0) = \frac{-a^2\left(\eta - \eta_1\right)}{a^2 - 1},
\end{align}
and 
\begin{align}
& L\left(\frac{2\eta}{\sigma^2}\right) = \notag \\  
& \begin{cases}
\frac{-2(2\eta - a)(\eta - \eta_2)(\eta - \eta_2')}{(a-2\eta)^2 - 1}, & \eta \in \left( 0, \frac{a-1}{2} \right)\cup \left( \frac{a+1}{2}, +\infty\right) \\
\frac{\eta}{1 - (a-2\eta)^2}, & \eta\in \left( \frac{a-1}{2} ,  \frac{a+1}{2}\right),
\end{cases}\label{eqn:Lright}
\end{align}
where $\eta_2'$ is given by 
\begin{align}
\eta_2'\triangleq \frac{3a - \sqrt{a^2 + 8}}{4}. \label{def:T2p}
\end{align}
Since $L(s)$ is an increasing function in $s$ due to~\eqref{eqn:doublederivative}, to determine the monotonicity of $t_1$, we only need to consider the following three cases. 

a) When $L(0)\geq 0$, or equivalently, $0 < \eta \leq \eta_1$, we have $L(s)\geq 0$ for any $s\in\mathcal{I}_\eta$. Hence, $t_1$ is an increasing function in $s$. 

b) When $L\left(\frac{2\eta}{\sigma^2}\right) \leq 0$, we have $L(s)\leq 0$ for any $s\in\mathcal{I}_\eta$. Hence, $t_1$ is a decreasing function in $s$. We now show that $L\left(\frac{2\eta}{\sigma^2}\right) \leq 0$ is equivalent to $\eta \geq \eta_2$. When $\eta \in \left(\frac{a-1}{2}, \frac{a+1}{2}\right)$, we have $L\left(\frac{2\eta}{\sigma^2}\right) > 0$ by~\eqref{eqn:Lright} and $\eta > 0$. When $\eta \in \left( 0, \frac{a-1}{2} \right)\cup \left( \frac{a+1}{2}, +\infty\right)$, it is easy to see from~\eqref{eqn:Lright} that  $L\left(\frac{2\eta}{\sigma^2}\right) \leq 0$ is equivalent to $\eta \in [\eta_2', a/2]\cup [\eta_2,+\infty)$. Hence, the equivalent condition for $L\left(\frac{2\eta}{\sigma^2}\right) \leq 0$ is $\eta \in [\eta_2,+\infty)$.  

c) When $L(0) < 0$ and $L\left(\frac{2\eta}{\sigma^2}\right) > 0$, or equivalently, $\eta \in (\eta_1, \eta_2)$, solving~\eqref{eqn:dt1} using~\eqref{eqn:firstderivative} yields~\eqref{def:sstar}. Since $L(s)$ is monotonically increasing due to~\eqref{eqn:doublederivative}, we know that $s^\star$ given by~\eqref{def:sstar} is the unique solution to~\eqref{eqn:dt1} in $\mathcal{I}_\eta$, and $L(s)\leq 0$ for $s\in (0, s^\star]$ and $L(s) > 0$ for $s\in (s^\star, 2\eta / \sigma^2 )$.
\end{proof}

\subsection{Proof of Lemma~\ref{lemma:limiting}}
\label{app:pfLemLimiting}
\begin{proof}
We first show the monotone decreasing property. The set $\mathcal{S}_{n+1}^{+}$ contains all $s>0$ such that $a_1,...,a_n, a_{n+1}$ are all less than $1/2\sigma^2$, while the set $\mathcal{S}_{n}^{+}$ contains all $s>0$ such that $a_1,...,a_n$ are all less than $1/2\sigma^2$, hence $\mathcal{S}_{n+1}^{+} \subseteq \mathcal{S}_{n}^{+} $. The same argument yields the conclusion for $\mathcal{S}_n^-$. 

We then prove that $\mathcal{S}_{\infty}^+=\left(0, 2\eta / \sigma^2\right]$. Property A4 above in Appendix~\ref{app:seqA} implies that for any $0< s\leq 2\eta / \sigma^2$, we have $\alpha_{\ell} \leq 0<  \frac{1}{2\sigma^2}$. Hence $ \left(0, 2\eta / \sigma^2\right] \subseteq\mathcal{S}_n^+$ for any $n\geq 1$. To show the other direction, it suffices to show that for any $s > \frac{2\eta}{\sigma^2}$, there exists $n\in\mathbb{N}$ such that $\alpha_n \geq \frac{1}{2\sigma^2}$. Let $\ell^\star$ be the integer defined in Property A4 above. Then, $\ell^\star$ satisfies the following two conditions
\begin{align}
&\frac{\alpha_1 - r_1}{\alpha_1 - r_2} q^{\ell^\star - 1} \geq  1, \label{ab1}\\
&\frac{\alpha_1 - r_1}{\alpha_1 - r_2} q^{\ell^\star} < 1.\label{bel1}
\end{align}
We show that $\alpha_{\ell^\star}\geq \frac{1}{2\sigma^2}$, which would complete the proof. Due to $r_2 - r_1 >0$, using~\eqref{eqn:expAlpha2} and~\eqref{bel1}, we have 
\begin{align}
\alpha_{\ell^\star} &\geq r_2 + \frac{r_2 - r_1}{\frac{1}{q} - 1} \\
&= \frac{r_2 - r_1 q}{1 - q} \label{eqn:2sigmapre} \\
&= \frac{1}{2\sigma^2}, \label{eqn:2sigma}
\end{align}
where~\eqref{eqn:2sigma}~\footnote{It is pretty amazing that~\eqref{eqn:2sigma} is in fact an equality.} is by plugging~\eqref{root1},~\eqref{root2} and~\eqref{cr} into~\eqref{eqn:2sigmapre}.  

Finally, to show~\eqref{minus}, for any $0<s\leq 2\eta / \sigma^2$, we have $\beta_{\ell} \leq 0 < \frac{1}{2\sigma^2},~\forall \ell \geq 1$, hence $\left(0, 2\eta / \sigma^2\right] \subseteq \mathcal{S}_{\infty}^-$. The other direction cannot hold since there are many counterexamples, e.g., $a = 1.2$, $\sigma^2 = 1$, $\eta = 0.15$ and $s = 0.35 > \frac{2\eta}{\sigma^2}$, where the sequence $\beta_\ell$ increases monotonically to $t_1 \approx 0.0411 < \frac{1}{2\sigma^2}$. Hence, in this case, $0.35 \in \mathcal{S}_{\infty}^-$ but $0.35\not\in \left(0, \frac{2\eta}{\sigma^2}\right]$.
\end{proof}

\subsection{Proof of Theorem~\ref{thm:upperLDP}}
\label{app:pfupperLDP}
\begin{proof}
Theorem~\ref{thm:chernoff} and Lemma~\ref{lemma:limiting} imply that for any $s\in\mathcal{I}_\eta$,
\begin{align}
\liminf_{n\rightarrow \infty} P^+(n, a, \eta) \geq \lim_{n\rightarrow \infty} \frac{1}{2n}\sum_{\ell = 1}^{n-1}\log (1 - 2\sigma^2\alpha_{\ell}).
\end{align}
Recall that $\alpha_{\ell}$ depends on $s$. By~\eqref{seqLim}, the continuity of the function $x\mapsto \log (1 - x)$ and the Ces{\`a}ro mean convergence, we have 
\begin{align}
\liminf_{n\rightarrow \infty} P^+(n, a, \eta)  \geq \frac{1}{2}\log(1 - 2\sigma^2 r_1), \label{pf:fixeds}
\end{align}
where $r_1$ depends on $s$ via~\eqref{root1}. Since~\eqref{pf:fixeds} holds for any $s\in\mathcal{I}_\eta$, using Property A5 in Appendix~\ref{app:seqA} above and supremizing~\eqref{pf:fixeds} over $s\in\mathcal{I}_\eta$, we obtain~\eqref{eqn:PplusL}. Specifically, the supremum of \eqref{pf:fixeds} over $s\in\mathcal{I}_\eta$ is achieved in the limit of $s$ going to the right end point $2\eta / \sigma^2$. Plugging $s = 2\eta / \sigma^2$ into~\eqref{root1}, we obtain the corresponding value for $r_1$:
\begin{align}
-\frac{(a+2\eta)^2 - 1}{2\sigma^2},
\end{align}
which is further substituted into~\eqref{pf:fixeds} to yield~\eqref{eqn:PplusL}.

Similarly, to show~\eqref{eqn:PminusL}, using Property B3 in Appendix~\ref{app:seqB} above, we have 
\begin{align}
\liminf_{n\rightarrow \infty} P^-(n, a, \eta)  \geq \sup_{s\in\mathcal{I}_\eta}\frac{1}{2}\log(1 - 2\sigma^2 t_1). \label{fixedss}
\end{align}
Then, by Property B4 in Appendix~\ref{app:seqB} above, the supermizer $s'$ in~\eqref{fixedss} is given by 
\begin{align}
s' = \begin{cases}
0, & 0< \eta \leq \eta_1\\
s^\star,& \eta_1 < \eta < \eta_2\\
\frac{2\eta}{\sigma^2}, & \eta \geq \eta_2,
\end{cases}
\label{super}
\end{align}
where $s^\star$ is given by~\eqref{def:sstar}. Plugging~\eqref{super} into~\eqref{fixedss} yields~\eqref{eqn:PminusL}.

Finally, the bound~\eqref{eqn:PL} follows from~\eqref{eqn:PplusL} and~\eqref{eqn:PminusL}, since 
\begin{align}
& \prob{\lrabs{\MLE-a}>\eta} \notag \\
=~& \prob{\LRB{\MLE-a}>\eta} + 
 \prob{\LRB{\MLE-a} < -\eta}
\end{align}
and 
\begin{align}
& \liminf_{n\rightarrow \infty} P(n, a, \eta)\notag \\
=~& \liminf_{n\rightarrow \infty}\min\left\{ P^+(n, a, \eta),  P^-(n, a, \eta)\right\} \\
 \geq ~& I^-(a, \eta).
\end{align}
\end{proof}

\subsection{Proof of Theorem~\ref{thm:decreasingeta}}
\label{app:pfdecreasingeta}
\begin{proof}
For any sequence $\eta_n$, the proof of Theorem~\ref{thm:chernoff} in Appendix~\ref{app:pfThmChernoff} above remains valid with $\alpha_{\ell}$ replaced by $\alpha_{n, \ell}$ defined in~\eqref{alphaElln} in Section~\ref{subsec:apade} above. We present the proof of~\eqref{eqn:decreasingBDp}, and omit that of~\eqref{eqn:decreasingBDm}, which is similar. In this regime, for each $n\geq 1$, the proof of Lemma~\ref{lemma:limiting} implies that 
\begin{align}
\LRB{0, \frac{2\eta_n}{\sigma^2}} \subseteq \mathcal{S}_n^+.
\end{align}
Then, in~\eqref{eqn:chernoffUpper}, we choose 
\begin{align}
s = s_n = \frac{\eta_n}{\sigma^2} \in \mathcal{S}_n^+. \label{choice:sn}
\end{align}
First, using~\eqref{root1}-\eqref{root2},~\eqref{cr} and the choice~\eqref{choice:sn}, we can determine the asymptotic behavior of quantities involved in determining $\alpha_{n, \ell}$ in~\eqref{eqn:expAlpha} and~\eqref{eqn:expAlpha2} (with $\eta$ replaced by $\eta_n$ and $s$ replaced by $s_n$), summarized in TABLE~\ref{table:dec}.

\begin{table}[h!]
\centering
 \begin{tabular}{||c | c | c | c | c | c ||} 
 \hline
   $\alpha_1$ & $r_1$ &  $r_2$ &  $r_2 - r_1$ & $q$ & $-\frac{\alpha_1 - r_1}{\alpha_1 - r_2}$  \\
 \hline
 $-\Theta(\eta_n^2)$ & 
 $-\Theta(1)$ & 
 $\Theta(\eta_n^2)$ & 
 $\Theta(1)$ & 
 $\Theta(1)$ &  
 $\Theta(1/\eta_n^2)$\\
 \hline
\end{tabular}
\caption{Order dependence in $\eta_n$ of the quantities involved in determining $\alpha_{n, \ell}$ in~\eqref{eqn:expAlpha} and~\eqref{eqn:expAlpha2}.}
\label{table:dec}
\end{table}
We make two remarks before proceeding further. It can be easily verified from~\eqref{cr} that the common ratio $q$ is a constant belonging to $(0,1)$ and 
\begin{align}
\lim_{\eta_n \rightarrow 0}~q = \frac{1}{a^2} \in (0, 1).
\end{align}
Hence, for all large $n$, $q$ is bounded by positive constants between 0 and 1. Besides, from~\eqref{root1}, we have 
\begin{align}
\lim_{\eta_n\rightarrow 0} r_1 = -\frac{a^2 - 1}{2\sigma^2}. \label{eqn:limR1}
\end{align}

Second, from~\eqref{eqn:expAlpha},~\eqref{eqn:chernoffUpper} and the choice~\eqref{choice:sn}, we have 
\begin{align}
& P^+(n, a, \eta_n)  \geq  \frac{n-1}{2n}\log\left (1 - 2\sigma^2 r_1\right ) +  \\
& \quad \frac{1}{2n}\sum_{\ell = 1}^{n-1}\log\LRB{1 - \frac{2\sigma^2 (r_2 - r_1)}{1 - 2\sigma^2 r_1}\cdot \frac{\LRB{-\frac{\alpha_1 - r_1}{\alpha_1 - r_2}}q^{\ell - 1}}{1 + \LRB{-\frac{\alpha_1 - r_1}{\alpha_1 - r_2}}q^{\ell - 1}}},\notag
\end{align}
where $r_1, r_2$ and $q$ in this regime depend on $\eta_n$ with order dependence given in TABLE~\ref{table:dec} above. Using the inequality $\log(1 - x) \geq \frac{x}{x - 1},~\forall x\in (0,1)$, we have 
\begin{align}
& P^+(n, a, \eta_n)  \geq  \frac{n-1}{2n}\log\left (1 - 2\sigma^2 r_1\right ) +  \\
&\quad\frac{1}{2n}\sum_{\ell = 1}^{n-1}\frac{-1}{\frac{1-2\sigma^2 r_2}{2\sigma^2 (r_2 - r_1)} + \frac{1-2\sigma^2 r_1}{2\sigma^2 (r_2 - r_1)} \cdot \frac{1}{\LRB{-\frac{\alpha_1 - r_1}{\alpha_1 - r_2}}q^{\ell - 1}}}.\notag
\end{align}
Since $1 - 2\sigma^2 r_2 >0$ due to~\eqref{root2}, we can further bound $P^+(n, a, \eta_n)$ as 
\begin{align}
 P^+(n, a, \eta_n)  & \geq  \frac{n-1}{2n}\log(1 - 2\sigma^2 r_1) -   \\
&~ \quad  \frac{1}{n}\LRB{\sum_{\ell = 1}^{n-1} q^{\ell - 1}} \frac{2\sigma^2 (r_2 - r_1)}{1 - 2\sigma^2 r_1}\cdot \LRB{-\frac{\alpha_1 - r_1}{\alpha_1 - r_2}} \notag \\
& \geq  \frac{n-1}{2n}\log(1 - 2\sigma^2 r_1) -   \\
&~ \quad  \frac{1}{n} \frac{2\sigma^2 (r_2 - r_1)}{(1 - 2\sigma^2 r_1)(1-q)}\cdot \LRB{-\frac{\alpha_1 - r_1}{\alpha_1 - r_2}} \notag \\
&~ = \frac{n-1}{2n}\log(1 - 2\sigma^2 r_1) - \frac{1}{n\Theta(\eta_n^2)},
\end{align}
where in the last step we used the results in TABLE~\ref{table:dec}. Due to the assumption~\eqref{assumption:etan} on $\eta_n$ and~\eqref{eqn:limR1}, we obtain~\eqref{eqn:decreasingBDp}.
\end{proof}

\subsection{Proof of Theorem~\ref{thm:subgaussian}}
\label{app:subG}
\begin{proof}
We point out the proof changes in generalizing our results to the sub-Gaussian case. There are two changes to be made in the proof of Theorem~\ref{thm:chernoff} in Appendix~\ref{app:pfThmChernoff} above: the equality from~\eqref{steps:0} to~\eqref{steps:1} is replaced by $\leq$ since $Z_n$ is $\sigma$-sub-Gaussian; the equality in~\eqref{diff} is replaced by $\leq$ due to~Lemma~\ref{lem:subG}. The rest of the proof for Theorem~\ref{thm:chernoff} remains the same for the sub-Gaussian case. Since Lemma~\ref{lemma:limiting} and Theorems~\ref{thm:upperLDP},~\ref{thm:decreasingeta} depend only on the properties of the sequences $\alpha_{\ell}$ and $\beta_{\ell}$, and~\eqref{eqn:chernoffUpper}-\eqref{eqn:chernoffLower} continue to hold for sub-Gaussian $Z_n$'s, the proofs of Lemma~\ref{lemma:limiting} and Theorems~\ref{thm:upperLDP},~\ref{thm:decreasingeta} remain exactly the same for the sub-Gaussian case.
\end{proof}

\section{}
\label{app:PfNRDF}
\subsection{Proof of Lemma~\ref{lemma:infdis}}
\label{app:LemInfdis}
\begin{proof}
In view of~\eqref{eqn:Xi}, we take the variances of both sides of~\eqref{eqn:dtiexp} to obtain
\begin{align}
\mathbb{V}_U(d) = \limsup_{n\rightarrow \infty}~\frac{1}{2n}\sum_{i=1}^n\min\left[1,~\left(\frac{\sigma_{n,i}^2}{\theta_n}\right)^2\right]. 
\label{pfeqn:VD}
\end{align}
Note that $\lim_{n\rightarrow\infty}\theta_n = \theta$, where $\theta > 0$ is the water level given by~\eqref{eqn:RWD}. Applying Theorem~\ref{thm:LimitingThm} in Section~\ref{subsubsec:diff} to~\eqref{pfeqn:VD} with the function 
\begin{align}
F(t) \triangleq \frac{1}{2} \min\left[1,~\left(\frac{\sigma^2}{\theta t}\right)^2\right],
\label{pfeqn:Ft}
\end{align}
which is continuous at $t = 0$, we obtain~\eqref{eqn:dispersion}.
\end{proof}

\subsection{An Integral}
\label{app:TwoIntegrals}
We present the computation of an interesting integral that is useful in obtaining the value of $\rdf_{U}(\dmax)$.
\begin{lemma}
\label{lemma:TwoIntegrals}
For any constant $r \in [-1,1]$, it holds that 
\begin{align}
\int_{-\pi}^{\pi} \log (1 - r\cos(w))~dw &= 4\pi \log \frac{\sqrt{1 + r} + \sqrt{1 - r}}{2}. \label{eqn:1stIntegral}
\end{align}
\end{lemma}

\begin{proof}
Denote 
\begin{align}
I(r)\delequal \int_{-\pi}^{\pi} \log (1 - r\cos(w))~dw.
\end{align}
By Leibniz's rule for differentiation under the integral sign, we have 
\begin{align}
\frac{\mathrm{d} I(r)}{\mathrm{d}r} &= \int_{-\pi}^{\pi} \frac{\partial }{\partial r}\log (1 - r\cos(w))~dw \\
&= -2\cdot \int_{0}^{\pi} \frac{\cos w}{1 - r\cos w}~dw. \label{int_derivative}
\end{align}
With the change of variable $u = \tan\left(w / 2\right)$ and partial-fraction decomposition, we obtain the closed-form solution to the integral in~\eqref{int_derivative}:
\begin{align}
\frac{\mathrm{d} I(r)}{\mathrm{d}r} = \frac{2\pi}{r}  - \frac{2\pi}{r\sqrt{1 - r^2}}. \label{eqn:dIR}
\end{align}
It can be easily verified by directly taking derivatives that the right-side of~\eqref{eqn:1stIntegral} is indeed the antiderivative of~\eqref{eqn:dIR}.
\end{proof}

\subsection{Derivation of $\rdf_{U}(\dmax)$ in~\eqref{rdfdmax}}
\label{app:rudmax}
We present two ways to obtain~\eqref{rdfdmax}. The first one is to directly use~\eqref{rdfHA} in Section~\ref{subsubsec:diff}. For $\theta = \theta_{\max}$, we have $\rdf_{\text{K}}(\dmax) = 0$ in~\eqref{Kol}, then~\eqref{rdfdmax} immediately follows from~\eqref{rdfHA}. The second method relies on~\eqref{eqn:RWR}. For $\theta = \theta_{\max}$, observe from~\eqref{eqn:RWR} that
\begin{align}
\rdf_{U}(\dmax) = \frac{1}{4\pi}\int_{-\pi}^{\pi}\log(g(w))~dw. \label{rdfdmaxint}
\end{align}
Then, computing the integral~\eqref{rdfdmaxint} using~Lemma~\ref{lemma:TwoIntegrals} in Appendix~\ref{app:TwoIntegrals} yields~\eqref{rdfdmax}.

\subsection{Proof of Lemma~\ref{lemma:eigScaling}}
\label{app:proof_eigScaling}
\begin{proof}
The bound~\eqref{eqn:eig2N} is obtained by partitioning $\mtx{F}' \mtx{F}$~\eqref{def:A} into its leading principal submatrix of order $n-1$ and then applying the Cauchy interlacing theorem to that partition, see~\cite[Lem. 1]{dispersionJournal} for details. To obtain~\eqref{eqn:eig1}, observe from~\eqref{eqn:prodMu}
\begin{align}
\mu_{n,1} = \left(\prod_{i = 2}^n \mu_{n,i}\right)^{-1}. \label{eqn:inv}
\end{align}
Combining~\eqref{eqn:inv} and~\eqref{eqn:eig2N} yields
\begin{align}
L_n \geq -\frac{1}{n}\log \mu_{n,1} \geq R_n,  \label{eqn:twosums}
\end{align}
where 
\begin{align}
L_n\delequal \frac{1}{n}\sum_{i = 2}^{n} \log \xi_{n, i} \quad \text{and}\quad R_n \delequal \frac{1}{n}\sum_{i = 1}^{n-1} \log \xi_{n-1, i}. \label{LnRn}
\end{align}
Plugging~\eqref{eqn:xini} into~\eqref{LnRn} and then taking the limit, we obtain
\begin{align}
\lim_{n\rightarrow\infty}L_n &= \lim_{n\rightarrow\infty} R_n \notag\\
&=  \frac{1}{\pi}\int_{0}^{\pi} \log (1 + a^2 - 2a\cos(w))~dw \\
&= 2\log a, \label{limsum}
\end{align}
where the last equality is due to Lemma~\ref{lemma:TwoIntegrals} in Appendix~\ref{app:TwoIntegrals} above. In the rest of the proof, we obtain the following refinement of~\eqref{limsum}: for any $n\geq 1$, 
\begin{align}
R_n &\geq 2\log a - \frac{c_1}{n}, \label{Rn}\\
L_n &\leq   2\log a + \frac{c_2}{n}, \label{Ln}
\end{align} 
where $c_1$ and $c_2$ are the constants given by~\eqref{int:c1} and~\eqref{int:c2} in Lemma~\ref{lemma:eigScaling}, respectively. Then,~\eqref{eqn:eig1} will follow directly from~\eqref{eqn:twosums},~\eqref{Rn} and~\eqref{Ln}. 

The proofs of the refinements~\eqref{Rn} and~\eqref{Ln} are similar, and both are based on the elementary relations between Riemann sums and their corresponding integrals. We present the proof of~\eqref{Rn}, and omit that of~\eqref{Ln}. Note that the function $h(w)\delequal \frac{1}{\pi}\log (1 + a^2 - 2a\cos(w))$ is an increasing function in $w\in [0, \pi]$, and its derivative is bounded above by $M_1 \delequal \frac{2a}{\pi (a^2-1)}$ for any fixed $a > 1$. Therefore, from~\eqref{eqn:xini} and~\eqref{LnRn}, we have 
\begin{align}
\left|R_n + \frac{1}{n}\log(a+1)^2 - \frac{1}{\pi}\int_{0}^{\pi} \log (g(w))~dw\right| \leq \frac{M_1\pi^2}{2n},
\end{align}
and~\eqref{Rn} follows immediately.
\end{proof}

\subsection{Proof of Theorem~\ref{thm:nonasymEig}}
\label{app:nonasymEig}
\begin{proof}
From Lemma~\ref{lemma:eigScaling}, we know that $\alpha' = 0 < \alpha $ (recall~\eqref{ma} and~\eqref{mprime}). Since $g(w)$ is an even function, we have 
\begin{align}
I & \delequal  \frac{1}{2\pi} \int_{-\pi}^{\pi} F(g(w))~dw  \\
& = \frac{1}{\pi} \int_{0}^{\pi} F(g(w))~dw.
\end{align}
Denote the maximum absolute value of $F$ over the interval~\eqref{interval_t} by $T>0$. It is easy to check that the function $F(g(w))$ is $2aL$-Lipschitz since $F(\cdot)$ is $L$-Lipschitz and the derivative of $g(w) $ is bounded by $2a$.  For the following Riemann sum
\begin{align}
S_n \delequal \frac{1}{n}\sum_{i = 1}^n F\left(g\left(\frac{i\pi}{n}\right)\right),
\end{align}
the Lipschitz property implies that 
\begin{align}
\lrabs{S_n - I} \leq \frac{2aL}{\pi n}.
\end{align}
For $i\geq 2$, rewrite~\eqref{eqn:xini} and~\eqref{eqn:eig2N} as 
\begin{align}
g\LRB{\frac{(i-1)\pi}{n}} \leq \mu_{n,i} \leq g\LRB{\frac{i\pi}{n+1}}. \label{pfeqn:mu}
\end{align}
Denote the sum in~\eqref{eqn:nonasymEig} as
\begin{align}
Q_n \triangleq \frac{1}{n}\sum_{i=1}^n F(\mu_{n,i}).
\end{align}
Then, separating $F(\mu_{n,1})$ from $Q_n$ and applying~\eqref{pfeqn:mu}, we have 
\begin{align}
Q_n &\geq S_n - \frac{2T}{n}, \\
Q_n & \leq \frac{n+1}{n} S_{n+1} + \frac{3T}{n}.
\end{align}
Therefore, there is a constant $C_L>0$ depending on $L$ and $T$ such that~\eqref{eqn:nonasymEig} holds.
\end{proof}

\section{}
\label{app:pfDispersion}
We gather the frequently used notations in this section as follows. For any given distortion threshold $d>0$,
\begin{itemize}
\item let $\theta>0$ be the water level corresponding to $d$ in the limiting reverse waterfilling~\eqref{eqn:RWD};

\item for each $n\geq 1$, let $\theta_n$ be the water level corresponding to $d$ in the $n$-th order reverse waterfilling~\eqref{eqn:NRDF2};

\item let $d_n$ be the distortion associated to the water level $\theta$ in the $n$-th order reverse waterfilling~\eqref{eqn:NRDF2}.
\end{itemize}
For clarity, we explicitly write down the relations between $d$ and $\theta_n$, and between $d_n$ and $\theta$:
\begin{align}
d &= \frac{1}{n}\sum_{i=1}^n \min(\theta_n,~\sigma_{n,i}^2), \label{pfeqn:thetaprime}\\
d_n &= \frac{1}{n}\sum_{i=1}^n \min(\theta,~\sigma_{n,i}^2),
\label{pfeqn:dprime}
\end{align} 
where $\sigma_{n,i}^2$'s are given in~\eqref{eqn:sigmai}. Note that $d$ and $\theta$ are constants independent of $n$, while $d_n$ and $\theta_n$ are functions of $n$, and there is no direct reverse waterfilling relation between $d_n$ and $\theta_n$. Applying Theorem~\ref{thm:LimitingThm} in Section~\ref{subsubsec:diff} above to the function $t\mapsto \min(\theta, \sigma^2 / t)$, we have 
\begin{align}
\lim_{n\rightarrow\infty} d_n = d, \label{eqn:ddn}
\end{align}
and 
\begin{align}
\lim_{n\rightarrow\infty}\theta_n = \theta.\label{eqn:thetathetan}
\end{align}
Theorem~\ref{thm:nonasymEig} in Section~\ref{subsec:spectrum} then implies that the speed of convergence in~\eqref{eqn:ddn} and~\eqref{eqn:thetathetan} is in the order of $1/n$.

\subsection{Expectation and Variance of the $\mathsf{d}$-tilted Information}
\label{app:EV}
\begin{proposition}
\label{prop:approx}
For any $d\in (0, \dmax)$ and $n\geq 1$, let $d_n$ be defined in~\eqref{pfeqn:dprime} above. Then, the expectation and variance of the $\mathsf{d}$-tilted information $\jmath_{\bfU}(\bfU, d_n)$ at distortion level $d_n$ satisfy
\begin{align}
\lrabs{ \frac{1}{n} \mathbb{E}\left[\jmath_{\bfU}(\bfU, d_n)\right] - \rdf_{U}(d) } &\leq \frac{C_1}{n},\label{eqn:appEXP}\\
\lrabs{ \frac{1}{n} \mathbb{V}\left[\jmath_{\bfU}(\bfU, d_n)\right] - \mathbb{V}_{U}(d) } &\leq \frac{C_2}{n},\label{eqn:appVAR}
\end{align}
where $\rdf_{U}(d)$ is the rate-distortion function given in~\eqref{eqn:RWR},  $\mathbb{V}_{U}(d)$ is the informational dispersion given in~\eqref{eqn:dispersion} and $C_1$, $C_2$ are positive constants.
\end{proposition}

\begin{proof}
Using the same derivation as that of~\eqref{eqn:dtiexp}, one can obtain the following representation of the $\mathsf{d}$-tilted information $\jmath_{\bfU}(\bfU, d_n)$ at distortion level $d_n$: 
\begin{align}
\jmath_{\bfU}(\bfU, d_n) &= \sum_{i = 1}^n \frac{\min(\theta,~\sigma_{n, i}^2)}{2\theta}\LRB{\frac{X_i^2}{\sigma_{n,i}^2} - 1} + \notag \\
& \frac{1}{2}\sum_{i=1}^n \log \frac{\max(\theta,~\sigma_{n,i}^2)}{\theta},
\label{pfeqn:DTId}
\end{align}
where $X_1^n$ is the decorrelation of $U_1^n$ defined in~\eqref{decor}. Note that the difference between~\eqref{eqn:dtiexp} and~\eqref{pfeqn:DTId} is that $\theta_n$ is replaced by $\theta$. Using~\eqref{eqn:Xi} and taking expectations and variances of both sides of~\eqref{pfeqn:DTId}, we arrive at
\begin{align}
\frac{1}{n} \mathbb{E}\left[\jmath_{\bfU}(\bfU, d_n)\right] &= \frac{1}{2n}\sum_{i=1}^n \log\max\left(1,~ \frac{\sigma_{n,i}^2}{\theta}\right), \label{pfeqn:expsum}\\
\frac{1}{n} \var{\jmath_{\bfU}(\bfU, d_n)} &= \frac{1}{2n}\sum_{i=1}^n \min\left(1,~\frac{\sigma_{n,i}^4}{\theta^2}\right) \label{pfeqn:varsum}.
\end{align}
Applying Theorem~\ref{thm:nonasymEig} in Section~\ref{subsec:spectrum} to~\eqref{pfeqn:expsum} with the function $F_{\text{G}}(t)$ defined in~\eqref{func:gray} yields~\eqref{eqn:appEXP}. Similarly, applying Theorem~\ref{thm:nonasymEig} to~\eqref{pfeqn:varsum} with the function~\eqref{pfeqn:Ft} yields~\eqref{eqn:appVAR}.
\end{proof}
Proposition~\ref{prop:approx} is one of the key lemmas that will be used in both converse and achievability proofs. Proposition~\ref{prop:approx} and its proof are similar to those of~\cite[Eq. (95)--(96)]{dispersionJournal}. The difference is that we apply Theorem~\ref{thm:nonasymEig}, which is the nonstationary version of~\cite[Th. 4]{dispersionJournal}, to a different function in \eqref{pfeqn:expsum}.

\subsection{Approximation of the $\mathsf{d}$-tilted Information}
\label{app:appDTI}
The following proposition gives a probabilistic characterization of the accuracy of approximating the $\mathsf{d}$-tilted information $\jmath_{\bfU} \left(\bfU, d \right)$ at distortion level $d$ using the $\mathsf{d}$-tilted information $\jmath_{\bfU} \left(\bfU, d_n \right)$ at distortion level $d_n$. 
\begin{proposition}
\label{prop:gapDTI}
For any $d \in (0, \dmax)$, there exists a constant $\tau>0$ (depending on $d$ only) such that for all $n$ large enough 
\begin{align}
\mathbb{P}\left[\left|\jmath_{\bfU} \left(\bfU, d \right) - \jmath_{\bfU} \left(\bfU, d_n \right)\right| > \tau \right] \leq \frac{1}{n},
\end{align}
where $d_n$ is defined in~\eqref{pfeqn:dprime}.
\end{proposition}

\begin{proof}
The proof in~\cite[App. D-B]{dispersionJournal} works for the nonstationary case as well, since the proof~\cite[App. D-B]{dispersionJournal} only relies on the convergences in~\eqref{eqn:ddn} and~\eqref{eqn:thetathetan} being both in the order of $1/n$, which continues to hold for the nonstationary case.
\end{proof}

\begin{remark}
The following high probability set is used in our converse and achievability proofs: 
\begin{align}
\mathcal{A} \triangleq \lbpara{\left|\jmath_{\bfU} \left(\bfU, d \right) - \jmath_{\bfU} \left(\bfU, d_n \right)\right| \leq  \tau}. 
\end{align}
Proposition~\ref{prop:gapDTI} implies that $\mathbb{P}[\mathcal{A}]\geq 1 - 1 / n$ for all $n$ large enough.
\label{rem:setE}
\end{remark}

\section{Converse Proof}
\label{app:pfConverse}
\begin{proof}[Proof of Theorem~\ref{thm:converse}]
Using the general converse by Kostina and Verd{\'u}~\cite[Th. 7]{kostina2012fixed} and our established Propositions~\ref{prop:approx} and~\ref{prop:gapDTI} in Appendix~\ref{app:pfDispersion}, the proof is the same as the converse proof in the asymptotically stationary case~\cite[Th. 7, Eq. (97)--(109)]{dispersionJournal}. For completeness, we give a proof sketch. Choosing $\gamma = ( \log n ) / 2$ and setting $X$ to be $U_1^n$ in~\cite[Th. 7]{kostina2012fixed}, we know that any $(n, M, d, \epsilon)$ code for the Gauss-Markov source must satisfy 
\begin{align}
\epsilon &\geq \prob{\jmath_{U_1^n}(U_1^n, d)\geq \log M +  ( \log n ) / 2 } - \frac{1}{\sqrt{n}}.
\end{align}
By conditioning on the high probability set $\mathcal{A}$ defined in~Remark~\ref{rem:setE} above, we can further bound $\epsilon$ from below by
\begin{align}
 & \left(1-\frac{1}{n}\right )\cdot \prob{\jmath_{U_1^n}(U_1^n, d_n)\geq \log M +  (\log n ) / 2 + \tau }  \notag \\
 & \qquad\qquad\qquad  \qquad\qquad\qquad  \qquad\qquad\qquad  - \frac{1}{\sqrt{n}}.
\end{align}
From~\eqref{pfeqn:DTId}, we know that $\jmath_{U_1^n}(U_1^n, d_n)$ is a sum of independent random variables, whose mean and variance are bounded (within the order of $1/n$ due to Proposition~\ref{prop:approx}) by the rate-distortion function $\R_{U}(d)$ and the informational dispersion $\mathbb{V}_U(d)$. Choosing $M$ as in~\cite[Eq. (103)]{dispersionJournal} and applying the Berry-Esseen theorem to $\jmath_{U_1^n}(U_1^n, d_n)$, we obtain the converse in Theorem~\ref{thm:converse}.
\end{proof}

\section{Achievability Proof}
\label{app:pfAchievability}
\begin{proof}[Proof of Theorem~\ref{thm:achievability}]
With our lossy AEP for the nonstationary Gauss-Markov source and Propositions~\ref{prop:approx} and~\ref{prop:gapDTI}, the proof is similar to the one for the stationary Gauss-Markov source in~\cite[Sec. V-C]{dispersionJournal}. Here, we streamline the proof. As elucidated in Section~\ref{subsec:dispersion} above, the standard random coding argument~\cite[Cor. 11]{kostina2012fixed} implies that for any $n$, there exists an $(n, M, d, \epsilon')$ code such that 
\begin{align}
\epsilon'\leq \inf_{P_{V_1^n}}~\mathbb{E}\lbrac{\exp\lpara{-M\cdot P_{V_{1}^n}(\mathcal{B}(U_1^n, d))}}.
\label{app:rcbound}
\end{align}
Choosing $V_1^n$ to be $V_1^{\star n}$ (the random variable that attains the minimum in~\eqref{eqn:nRDF} with $X_1^n$ there replaced by $U_1^n$), the bound~\eqref{app:rcbound} can be relaxed to
\begin{align}
\epsilon'\leq \mathbb{E}\lbrac{\exp\lpara{-M\cdot P_{V_{1}^{\star n}}(\mathcal{B}(U_1^n, d))}}.
\end{align}
To simplify notations, in the following, we denote by $C$ a constant that might be different from line to line. Given any constant $\epsilon\in (0,1)$, define $\epsilon_n$ as 
\begin{align}
\epsilon_n \triangleq \epsilon - \frac{C}{\sqrt{n}} - \frac{1}{q(n)}  - \frac{1}{n}, \label{app:en}
\end{align}
where $q(n)$ is defined in~\eqref{eqn:qn} above. Note that for all $n$ large enough, we have $\epsilon_n\in (0, 1)$. We choose $M$ as 
\begin{align}
\log M\triangleq~& n\R_{U}(d) + \sqrt{n \mathbb{V}_U(d)} Q^{-1}(\epsilon_n) + \notag \\
&\quad \log(\log n / 2) + p(n) + C + \tau, \label{app:logm}
\end{align}
where $p(n)$ is defined in~\eqref{eqn:pn} and $\tau$ is from Proposition~\ref{prop:gapDTI} above. We also define the random variable $G_n$ as 
\begin{align}
G_n \triangleq \log M - \jmath_{U_1^n}(U_1^n, d_n) - p(n) - C - \tau,
\end{align}
where $d_n$ is defined in~\eqref{pfeqn:dprime} above. Note that all the randomness in $G_n$ is from $U_1^n$, hence we will also use the notation $G_n(u_1^n)$ to indicate one realization of the random variable $G_n$. By bounding the deterministic part, that is, $\log M$,   of $G_n$ using Proposition~\ref{prop:approx}, we know that with probability 1, 
\begin{align}
G_n\geq \mathbb{E} + Q^{-1}(\epsilon_n)\sqrt{\mathbb{V}} - \jmath_{U_1^n}(U_1^n, d_n) + \log(\log n / 2),
\end{align}
where we use $\mathbb{E}$ and $\mathbb{V}$ to denote the expectation and variance of the informational dispersion $\jmath_{U_1^n}(U_1^n, d_n)$ at distortion level $d_n$. Define the set $\mathcal{G}_n$ as 
\begin{align}
\mathcal{G}_n \triangleq \lbpara{u_1^n\in\R^n\colon G_n(u_1^n) < \log(\log n / 2) }, 
\end{align}
Then, in view of~\eqref{pfeqn:DTId}, the informational dispersion $\jmath_{U_1^n}(U_1^n, d_n)$ is a sum of independent random variables with bounded moments, and we apply the Berry-Esseen theorem to obtain 
\begin{align}
P_{U_1^n}(\mathcal{G}_n) \leq \epsilon_n +\frac{C}{\sqrt{n}}.
\end{align}
We define one more set $\mathcal{L}_n$ as 
\begin{align}
\mathcal{L}_n \triangleq \lbpara{u_1^n\in\R^n\colon \log\frac{1}{P_{V_1^{\star n}}(\mathcal{B}(u_1^n, d))} < \log M - G_n(u_1^n) }.
\end{align}
Then, by the lossy AEP in Lemma~\ref{lemma:LossyAEP} in Section~\ref{subsec:dispersion} above and Proposition~\ref{prop:gapDTI}, we have 
\begin{align}
P_{U_1^n}(\mathcal{L}_n) \geq 1 - \frac{1}{q(n)} - \frac{1}{n}. \label{app:proln}
\end{align}
Finally, for any constant $\epsilon\in (0,1)$ and $n$ large enough, we define $\epsilon_n$ as in~\eqref{app:en} above and set $M$ as in~\eqref{app:logm}. Then, there exists $(n, M, d, \epsilon')$ code such that  
\begin{align}
 \epsilon' \leq &~ \mathbb{E}\lbrac{\exp\lpara{-M\cdot P_{V_{1}^{\star n}}(\mathcal{B}(U_1^n, d))\cdot 1\{\mathcal{L}_n\}}} + \notag \\
&~\mathbb{E}\lbrac{\exp\lpara{-M\cdot P_{V_{1}^{\star n}}(\mathcal{B}(U_1^n, d))}\cdot 1\{\mathcal{L}_n^c\}} \\
\leq &~\mathbb{E}\lbrac{\exp(e^{-G_n})} + \frac{1}{q(n)} + \frac{1}{n},
\end{align}
where the last inequality is due to the definition of $\mathcal{L}_n$ and~\eqref{app:proln}. By further conditioning on $\mathcal{G}_n$, we conclude that there exists $(n, M, d, \epsilon')$ code such that  
\begin{align}
\epsilon' & \leq \epsilon_n + \frac{C}{\sqrt{n}} + \frac{1}{n} + \frac{1}{q(n)} \\
&= \epsilon.
\end{align}
Therefore, by the choice of $M$ in~\eqref{app:logm}, the minimum achievable source coding rate $R(n, d, \epsilon)$ must satisfy 
\begin{align}
R(n, d, \epsilon)  & \leq \mathbb{R}_{U} (d) + \sqrt{\frac{\mathbb{V}_U(d)}{n}} Q^{-1}(\epsilon) + \notag \\
 &\quad \frac{K_1 \log\log n}{n} + \frac{p(n)}{n} + \frac{K_2}{\sqrt{n} q(n)} \label{eqn:generalrelation}
\end{align} 
for all $n$ large enough, where $K_1> 0$ is a universal constant and $K_2$ is a constant depending on $\epsilon$. Here we change from $Q^{-1}(\epsilon_n)$ to $Q^{-1}(\epsilon)$ using a Taylor expansion. Therefore, Theorem~\ref{thm:achievability} follows immediately from~\eqref{eqn:generalrelation} with the choices of $p(n)$ and $q(n)$ given by~\eqref{eqn:pn} and~\eqref{eqn:qn}, respectively, in the lossy AEP in Lemma~\ref{lemma:LossyAEP} in Section~\ref{subsec:dispersion} above. We have $O(\cdot)$ in~\eqref{abound} since $K_2$ could be positive or negative. 
\end{proof}

\section{Proof of Lossy AEP}
\label{app:Achievability}

\subsection{Notations}
\label{app:lossyAEPnotations}
For the optimization problem $\rdf(A_1^n, B_1^n, d)$ in~\eqref{eqn:crem}, the generalized tilted information defined in~\cite[Eq. (28)]{kostina2012fixed} in $a_1^n$ (a realization of $A_1^n$) is given by
\begin{align}
\Lambda_{B_1^n}(a_1^n, \delta, d) \delequal -\delta n d - \log \EX{\exp(-n\delta \dis{a_1^n}{B_1^n})},
\end{align} 
where $\delta>0$ and $d\in(0,\dmax)$. For properties of the generalized tilted information, see~\cite[App. D]{kostina2012fixed}. For clarity, we list the notations used throughout this section:  
\begin{enumerate}
\item $\bfX$ denotes the decorrelation of $\bfU$ defined in~\eqref{decor};

\item $\hat{X}_1^n$ is the proxy random variable of $\bfX$ defined in Definition~\ref{def:proxy} in Section~\ref{subsec:LossyAEPandPE} above; 

\item For $Y_1^{\star n}$ that achieves $\rdf_{\bfX}(d)$ in~\eqref{eqn:nRDF}, $\hat{F}_1^{\star n}$ is the random vector that achieves $\rdf\LRB{\hat{X}_1^n, Y_1^{\star n}, d}$; 

\item We denote by $\lambda^\star_n$ the negative slope of $\rdf_{\bfX}(d)$ (the same notation used in~\eqref{dtiltedGM}): 
\begin{align}
\lambda^\star_n \triangleq -\rdf'_{\bfX}(d). \label{lambdaOpt}
\end{align}
It is shown in~\cite[Lem. 5]{dispersionJournal} that $\lambda^\star_n$ is related to the $n$-th order water level $\theta_n$ in~\eqref{eqn:NRDF2} by 
\begin{align}
\lambda^\star_n = \frac{1}{2\theta_n}. \label{eqn:lambdatheta}
\end{align}
Given any source outcome $\bfu$, let $x_1^n$ be the decorrelation of $\bfu$. Define $\hat{\lambda}_n$ as the negative slope of $\rdf(\hat{X}_1^n, Y_1^{\star n}, d)$ w.r.t. $d$:
\begin{align}
\hat{\lambda}_n \triangleq -\rdf'(\hat{X}_1^n, Y_1^{\star n}, d).\label{lambdaXOpt}
\end{align}

\item Comparing the definitions of $\mathsf{d}$-tilted information and the generalized tilted information, one can see that~\cite[Eq. (18)]{dispersionJournal}
\begin{align}
\jmath_{\bfX}(x_1^n, d) = \Lambda_{Y_1^{\star n}}(\bfx, \lambda^\star_n, d). 
\end{align}

\item Recalling~\eqref{eqn:Xi} and applying the reverse waterfilling result~\cite[Th. 10.3.3]{cover2012elements}, we know that the coordinates of $Y_1^{\star n}$ are independent and satisfy 
\begin{align}
Y_i^\star \sim \mathcal{N}(0,~\nu_{n,i}^2), 
\label{Yi}
\end{align}
where 
\begin{align}
\nu_{n,i}^2 \delequal \max(0,~\sigma_{n,i}^2 - \theta_n),
\label{nui}
\end{align}
with $\theta_n>0$ given in~\eqref{pfeqn:thetaprime}.
\end{enumerate}

\subsection{Parametric Representation of the Gaussian Conditional Relative Entropy Minimization}
\label{app:paraCREM}
Various aspects of the optimization problem~\eqref{eqn:crem} have been discussed in~\cite[Sec. II-B]{dispersionJournal}. In particular, let $B_1^{\star n}$ be the optimizer of $\R_{A_1^n}(d)$, then we have 
\begin{align}
\mathbb{R}(A_1^n, B_1^{\star n}, d) = \R_{A_1^n}(d),
\end{align}
where $\R_{A_1^n}(d)$ is in~\eqref{eqn:nRDF}. Another useful result on the optimization problem~\eqref{eqn:crem} is the following: when $A_1^n$ and $B_1^n$ are independent Gaussian random vectors, the next theorem gives parametric characterizations for the optimizer and optimal value of~\eqref{eqn:crem}. 

\begin{theorem}
\label{th:gaussiancrem}
Let $A_1,\ldots, A_n$ be independent random variables with 
\begin{align}
A_i\sim \mathcal{N}(0, \alpha_i^2),\label{eq:Aalphai}
\end{align}
and $B_1,\ldots, B_n$ be independent random variables with 
\begin{align}
B_i\sim \mathcal{N}(0, \beta_i^2),\label{eq:Bbetai}.
\end{align}
For any $d$ such that 
\begin{align}
0< d < \frac{1}{n}\sum_{i=1}^n (\alpha_i^2 + \beta_i^2), \label{eq:ranged}
\end{align}
we have the following parametric representation for $\R(A_1^n, B_1^n, d)$: 
\begin{align}
\R(A_1^n, B_1^n, d) & = -\lambda d ~+ \label{eq:rlambdadcrem}\\
&\quad \frac{1}{2n}\sum_{i=1}^n \log (1 + 2\lambda \beta_i^2) + 
 \frac{1}{n}\sum_{i=1}^n \frac{\lambda \alpha_i^2}{1 + 2\lambda \beta_i^2} \notag
\end{align}
\begin{align}
d &= \frac{1}{n}\sum_{i=1}^n \frac{\alpha_i^2 + \beta_i^2 (1 + 2\lambda \beta_i^2)}{(1 + 2\lambda \beta_i^2)^2},\label{eq:dlambdadcrem}
\end{align}
where $\lambda > 0$ is the parameter. Furthermore, $\lambda$ equals the negative slope of $\mathbb{R}(A_1^n, B_1^n, d)$ w.r.t. $d$:
\begin{align}
\lambda = -\mathbb{R}'(A_1^n, B_1^n, d). \label{app:llambda}
\end{align}
\end{theorem}

Similar results to Theorem~\ref{th:gaussiancrem} have appeared previously in the literature~\cite{dembo1999asymptotics,yang1999redundancy,dembo2002source}. See~\cite[Example 1 and Th. 2]{dembo2002source} for the case of $n=1$. For completeness, we present a proof.

\begin{proof}
Fix any $d$ that satisfies~\eqref{eq:ranged}, and let $\lambda$ be such that~\eqref{eq:dlambdadcrem} is satisfied. Note from~\eqref{eq:dlambdadcrem} that $d$ is a strictly decreasing function in $\lambda$ (unless $\beta_i = 0$ for all $i\in [n]$), hence such $\lambda$ is unique. The upper bound on $d$ in~\eqref{eq:ranged} guarantees that $\lambda > 0$. We first show the $\leq$ direction in~\eqref{eq:rlambdadcrem}. For $A_1^n = a_1^n\in\mathbb{R}^n$, define the conditional distribution $P_{F_i |A_i = a_i} (f_i)$ as
\begin{align}
\mathcal{N}\lpara{\frac{2\lambda \beta_{i}^2 a_i}{1+2\lambda \beta_{i}^2}, \frac{\beta_{i}^2}{1+2\lambda \beta_{i}^2}}.\label{eq:cupper}
\end{align} 
We then define the joint distribution $P_{A_1^n, F_1^n}$ as
\begin{align}
P_{A_1^n, F_1^n} \triangleq \prod_{i = 1}^n P_{F_i | A_i} P_{A_i}. \label{eq:prodchoice}
\end{align}
Using~\eqref{eq:dlambdadcrem}, we can check that with such a choice of $P_{A_1^n, F_1^n}$, the expected distortion between $A_1^n$ and $F_1^n$ equals $d$. The details follow. 
\begin{align}
\mathbb{E}\left[\dis{A_1^n}{F_1^n}\right]  = &~  \mathbb{E}\left[ \mathbb{E}[\dis{A_1^n}{F_1^n} | A_1^n]\right] \\
= &~\frac{1}{n}\sum_{i=1}^n \mathbb{E}\left[ \mathbb{E}[(F_i - A_i)^2| A_i] \right]\\
= &~ \frac{1}{n}\sum_{i=1}^n \frac{\beta_{i}^2}{1+2\lambda \beta_{i}^2} + \frac{\alpha_i^2}{(1 + 2\lambda\beta_i^2)^2} \label{ssstep:step1}\\
= &~ d, \label{ssstep:step2}
\end{align}
where~\eqref{ssstep:step1} is from the relation $\mathbb{E}[(X - t)^2] = \text{Var}[X] + (\mathbb{E}[X] - t)^2$ and~\eqref{ssstep:step2} is due to~\eqref{eq:dlambdadcrem}. Therefore, the choice of $P_{F_1^n|A_1^n}$ in~\eqref{eq:cupper} and~\eqref{eq:prodchoice} is feasible for the optimization problem in defining $\rdf(A_1^n, B_1^n, d)$. Hence, 
\begin{align}
\rdf(A_1^n, B_1^n, d) &\leq \frac{1}{n} D\left (P_{F_1^n | A_1^n} || P_{B_1^n} | P_{A_1^n}\right ) \\
&= \frac{1}{n} \sum_{i=1}^n \mathbb{E} \left [D\left (P_{F_i | A_i}(\cdot | A_i) || P_{B_i} \right )\right ].\label{eq:sumkl}
\end{align}
It is straightforward to verify that the Kullback-Leibler divergence between two Gaussian distributions $X\sim \mathcal{N}(\mu_X, \sigma_X^2)$ and $Y\sim \mathcal{N}(\mu_Y, \sigma_Y^2)$ is given by
\begin{align}
D(P_X||P_Y) =\frac{\sigma_X^2 + (\mu_X - \mu_Y)^2}{2\sigma_Y^2} -  \frac{1}{2}\log\frac{\sigma_X^2}{\sigma_Y^2} - \frac{1}{2}. \label{eq:klgaussians}
\end{align}
Using~\eqref{eq:klgaussians} and~\eqref{eq:cupper}, we see that~\eqref{eq:sumkl} equals the right-hand side of~\eqref{eq:rlambdadcrem}. To prove the other direction, we use the Donsker-Varadhan representation of the Kullback-Leibler divergence~\cite[Th. 3.5]{polyanskiy2014lecture}:
\begin{align}
D(P||Q) = \sup_{g}~\mathbb{E}_P[g(X)] - \log \mathbb{E}_Q[\exp{g(X)}],\label{eq:dv}
\end{align}
where the supremum is over all functions $g$ from the sample space to $\R$ such that both expectations in~\eqref{eq:dv} are finite. Fix any $P_{F_1^n|A_1^n}$ such that $\mathbb{E}[\dis{A_1^n}{F_1^n}] \leq d$. For any $A_1^n = a_1^n$, in~\eqref{eq:dv}, we choose $P$ to be $P_{F_1^n|A_1^n =  a_1^n}$, $Q$ to be $P_{B_1^n}$ and $g$ to be $g(f_1^n)\triangleq -n \lambda \mathsf{d}(f_1^n, a_1^n)$ for any $f_1^n\in\R^n$, then we have 
\begin{align}
D(P_{F_1^n|A_1^n = a_1^n} ||P_{B_1^n})  & \geq - n\lambda \mathbb{E}_{P_{F_1^n|A_1^n = a_1^n}}[\mathsf{d}(F_1^n, a_1^n)] \label{eq:dvgaussian}\\ &\quad - \log \mathbb{E}_{P_{B_1^n}}[\exp\left(-n\lambda \dis{B_1^n}{a_1^n}\right)]. \notag 
\end{align}
Taking expectations on both sides of~\eqref{eq:dvgaussian} with respect to $P_{A_1^n}$ and then normalizing by $n$, we have 
\begin{align}
\rdf(A_1^n, B_1^n, d)  &\geq -\lambda \mathbb{E}[\dis{A_1^n}{F_1^n}]  \label{eq:dvbound} \\ & \quad- \mathbb{E}_{P_{A_1^n}}\log \mathbb{E}_{P_{B_1^n}}[\exp\left(-n\lambda \dis{B_1^n}{A_1^n}\right)]. \notag
\end{align}
Using the formula for the moment generating function for noncentral $\chi^2$ distributions, we can compute 
\begin{align}
& \mathbb{E}_{P_{B_1^n}}[\exp\left(-n\lambda \dis{B_1^n}{a_1^n}\right)] \notag \\
= &~ \prod_{i = 1}^n \frac{1}{\sqrt{1 + 2\lambda\beta_i^2}} \exp\lpara{\frac{-\lambda a_i^2}{1 + 2\lambda\beta_i^2}}. \label{eq:mgfchisq}
\end{align}
Plugging~\eqref{eq:mgfchisq} into~\eqref{eq:dvbound} and using $\mathbb{E}[\dis{A_1^n}{F_1^n}] \leq d$, we conclude that $\rdf(A_1^n, B_1^n, d)$ is greater than or equal to the right-hand side of~\eqref{eq:rlambdadcrem}. Finally,~\eqref{app:llambda} is obtained by taking derivative of~\eqref{eq:rlambdadcrem} w.r.t. $d$, where we need to use the chain rule for derivatives since $\lambda$ is a function of $d$ given by~\eqref{eq:dlambdadcrem}.
\end{proof}

Our next result states that for fixed $\beta_i^2$'s satisfying certain mild conditions, if we change the variances from $\alpha_i^2$'s to $\hat{\alpha}_i^2$'s, then the perturbation on the corresponding $\lambda$'s is controlled by the perturbation on $\alpha_i^2$'s. 
\begin{theorem}[Variance perturbation]
\label{th:vp}
Let $\alpha_i^2$'s and $\beta_i^2$'s be in~\eqref{eq:Aalphai} and \eqref{eq:Bbetai} above, respectively. For a fixed $d$ satisfying~\eqref{eq:ranged}, let $\lambda$ be given by~\eqref{eq:dlambdadcrem}. Suppose that $\alpha_i^2$'s and $\beta_i^2$'s are such that both 
\begin{align}
\frac{1}{n}\sum_{i=1}^n\frac{1}{(1 + 2\lambda \beta_i^2)^4} \label{aspt:beta1}
\end{align} 
and 
\begin{align}
\frac{1}{n}\sum_{i=1}^n \frac{2\beta_i^2(2\alpha_i^2 + 1 + 2\lambda\beta_i^2)}{(1 + 2\lambda\beta_i^2)^3} \label{aspt:beta2}
\end{align} 
are bounded above by positive constants. Let $\hat{A}_1,\ldots, \hat{A}_n$ be independent random variables with 
\begin{align}
\hat{A}_i\sim \mathcal{N}(0, \hat{\alpha}_i^2).\label{eq:Ahatalphai}
\end{align}
Let $\hat{\lambda}$ be such that 
\begin{align}
d &= \frac{1}{n}\sum_{i=1}^n \frac{\hat{\alpha}_i^2 + \beta_i^2 (1 + 2\hat{\lambda} \beta_i^2)}{(1 + 2\hat{\lambda} \beta_i^2)^2}.\label{eq:dhatlambdadcrem}
\end{align}

Then, there is a constant $C>0$ such that  
\begin{align}
\abs{\hat{\lambda} - \lambda } \leq C \max_{1\leq i\leq n}~\abs{\hat{\alpha}_i^2 - \alpha_i^2}.
\end{align}
\end{theorem}

\begin{proof}
We can view~\eqref{eq:dlambdadcrem} as an equation of the form $f(\alpha_1^2, \ldots, \alpha_n^2, \lambda) = 0$. Then, by the implicit function theorem, we know that there exists a unique continuously differentiable function $h$ such that 
\begin{align}
\lambda = h(\alpha_1^2,\ldots,\alpha_n^2),
\end{align}
and
\begin{align}
\frac{\partial h }{\partial \alpha_i^2} = \lbpara{\frac{1}{n}\sum_{i=1}^n \frac{2\beta_i^2 [ 2\alpha_i^2 + \beta_i^2(1 + 2\lambda\beta_i^2)]}{(1 + 2\lambda\beta_i^2)^3}}^{-1} \frac{1}{n(1 + 2\lambda \beta_i^2)^2}.
\end{align} 
Hence, 
\begin{align}
\pnorm{2}{\nabla h}  = &   \lbpara{\frac{1}{n}\sum_{i=1}^n \frac{2\beta_i^2(2\alpha_i^2 + 1 + 2\lambda\beta_i^2)}{(1 + 2\lambda\beta_i^2)^3}}^{-1}   \times \\ 
& \quad\quad  \sqrt{\frac{1}{n^2}\sum_{i=1}^n\frac{1}{(1 + 2\lambda \beta_i^2)^4}}.\notag 
\end{align}
By the assumptions~\eqref{aspt:beta1} and~\eqref{aspt:beta2}, we know that there exists a constant $C>0$ such that 
\begin{align}
\pnorm{2}{\nabla h}  \leq \frac{C}{\sqrt{n}}.
\end{align}
Hence, we have 
\begin{align}
\abs{\hat{\lambda} - \lambda } & \leq \pnorm{2}{\nabla h} \pnorm{2}{(\alpha_1^2,\ldots, \alpha_n^2) - (\hat{\alpha}_1^2,\ldots, \hat{\alpha}_n^2)} \\
&\leq C \max_{1\leq i\leq n}~\abs{\hat{\alpha}_i^2 - \alpha_i^2}.
\end{align}
\end{proof}

\subsection{Proof of Theorem~\ref{thm:typical}}
\label{app:PfThTS}
The proof is similar to~\cite[Th. 12]{dispersionJournal}. We streamline the proof and point out the differences. We use the notations defined in Appendix~\ref{app:lossyAEPnotations} above.

Our Corollary~\ref{cor:disp} implies that for all $n$ large enough the condition~\eqref{eqn:cond1} is violated with probability at most $2e^{-cn}$ for a constant $c> \log (a) / 2$. This is much stronger than the bound $\Theta\left(1 / \text{poly}\log n \right)$ in the stationary case~\cite[Th. 6]{dispersionJournal}.

In view of~\eqref{eqn:Xi}, the random variables $X_i / \sigma_{n,i}$ for $i = 1,\ldots, n$, are distributed according to i.i.d. standard normal distributions, and their $2k$-th moments equal to $(2k-1)!!$. The Berry-Esseen theorem implies that the condition~\eqref{eqn:cond2} is violated with probability at most $\Theta \left(1 / \sqrt{n}\right)$. This is the same as in the stationary case~\cite[Eq. (279)--(280)]{dispersionJournal}.

We use the following procedure to show that the condition~\eqref{eqn:cond3} is violated with probability at most $\Theta\left(1/\log n\right)$: 
\begin{itemize}
\item We approximate $m_i(u_1^n)$ by another random variable $\bar{m}_{i}(u_1^n)$ that is easier to analyze. 
\item We show that~\eqref{eqn:cond3} with $m_{i}(u_1^n)$ replaced by $\bar{m}_i(u_1^n)$ holds with probability at least $1 - \Theta(1 / \log n)$. 
\item We then control the difference between $m_i(u_1^n)$ and $\bar{m}_i(u_1^n)$. 
\end{itemize}
To carry out the above program, we first give an expression for $m_i(u_1^n)$ by applying~\cite[Lem. 4]{dispersionJournal} (see also the proof of Theorem~\ref{th:gaussiancrem}) on $\R(\hat{X}_1^n, Y_1^{\star n}, d)$. Note that $\hat{X}_1^n$ and $Y_1^{\star n}$ are Gaussian random vectors with independent coordinates with variances given by~\eqref{eqn:sigmaihat} and~\eqref{Yi}, respectively. Then,~\cite[Lem. 4]{dispersionJournal} implies that the optimizer $P_{\hat{F}_1^{\star n} | \hat{X}_1^n}$ for $\R(\hat{X}_1^n, Y_1^{\star n}, d)$ satisfies 
\begin{align}
P_{\hat{F}_1^{\star n} | \hat{X}_1^n = \hat{x}_1^n}  = \prod_{i=1}^n P_{\hat{F}_i^{\star} | \hat{X}_i = \hat{x}_i},
\end{align}
where the conditional distributions $\hat{F}_i^{\star} | \hat{X}_i = \hat{x}_i$ are Gaussian: 
\begin{align}
\mathcal{N}\lpara{\frac{2\lambdahat \nu_{n,i}^2 \hat{x}_i}{1+2\lambdahat \nu_{n,i}^2}, \frac{\nu_{n,i}^2}{1+2\lambdahat \nu_{n,i}^2}}, \label{FiXi}
\end{align}
where $\nu_{n,i}^2$'s are defined in~\eqref{nui}, and $\lambdahat$ is defined in~\eqref{lambdaXOpt}. Then, using the definition of $m_i(u_1^n)$ in~\eqref{mi} and~\eqref{FiXi}, we obtain 
\begin{align}
m_i(u_1^n) = \frac{\nu_{n,i}^2}{1+2\lambdahat \nu_{n,i}^2} + \frac{x_i^2}{(1+2\lambdahat \nu_{n,i}^2)^2}, \label{eqn:miexp}
\end{align}
where $x_1^n = \mtx{S}' u_1^n$. The random variable $m_i(u_1^n)$ in the form of~\eqref{eqn:miexp} is hard to analyze since we do not have a simple expression for $\lambdahat$. By replacing $\lambdahat$ with $\lambda^\star_n$, we define another random variable $\bar{m}_{i}(u_1^n)$ that turns out to be easier to analyze:
\begin{align}
\bar{m}_i(u_1^n) \triangleq \frac{\nu_{n,i}^2}{1+2\lambda^\star_n \nu_{n,i}^2} + \frac{x_i^2}{(1+2\lambda^\star_n \nu_{n,i}^2)^2}. \label{eqn:mibarexp}
\end{align}
Plugging~\eqref{eqn:lambdatheta} and~\eqref{nui} into~\eqref{eqn:mibarexp}, we obtain
\begin{align}
\bar{m}_i(u_1^n)= \frac{\min(\sigma_{n,i}^2, \theta_n)^2}{\sigma_{n,i}^2}\lpara{\frac{x_i^2}{\sigma_{n,i}^2} - 1} + \min(\sigma_{n,i}^2, \theta_n), \label{eqn:mibarexp1}
\end{align}
where $\theta_n$ is the $n$-th order water level in~\eqref{eqn:NRDF2} and $x_1^n = \mtx{S}' u_1^n$. The random variable $\bar{m}_i(U_1^n)$ is much easier to analyze since $X_i / \sigma_{n,i}$'s are i.i.d. standard normal random variables. Moreover, in view of~\eqref{eqn:NRDF2}, their expectations satisfy
\begin{align}
\frac{1}{n}\sum_{i=1}^n \EX{\bar{m}_i(U_1^n)} = \frac{1}{n}\sum_{i=1}^n  \min(\sigma_{n,i}^2, \theta_n) = d.
\end{align}
Since $X_i / \sigma_{n,i}$ has bounded moments, from the Berry-Esseen theorem, we know that there exists a constant $\omega > 0$ such that for all $n$ large enough
\begin{align}
\prob{\abs{\frac{1}{n}\sum_{i=1}^n \bar{m}_i(U_1^n) - d} > \omega \eta_n} \leq \frac{C_1}{\log n} + \frac{C_2}{\sqrt{n}}, \label{mbarbound}
\end{align}
where $\eta_n$ is in~\eqref{eqn:etan} above, and $C_1, C_2$ are positive constants. In the last step of the program, we control the difference between $m_{i}(U_1^n)$ and $\bar{m}_i(U_1^n)$. From~\eqref{eqn:miexp}--\eqref{eqn:mibarexp}, we have 
\begin{align}
& \frac{1}{n}\sum_{i=1}^n \bar{m}_i(u_1^n) - \frac{1}{n}\sum_{i = 1}^n m_i(u_1^n) \notag\\
= & \frac{1}{n}\sum_{i = 1}^n \frac{2\nu_{n,i}^4 (\lambdahat - \lambda^\star_n)}{(1 + 2 \lambdahat \nu_{n,i}^2)(1 + 2\lambda^\star_n \nu_{n,i}^2)} ~+  \label{eqn:Modified} \\
&\quad \frac{1}{n}\sum_{i = 1}^n \frac{2x_i^2\nu_{n,i}^2 (2 + 2\lambdahat \nu_{n,i}^2 +2 \lambda^\star_n \nu_{n,i}^2)(\lambdahat - \lambda^\star_n)}{(1 + 2 \lambdahat \nu_{n,i}^2)^2(1 + 2\lambda^\star_n \nu_{n,i}^2)^2}.\notag
\end{align}
For $i=1$, we have $\nu_{n,1}^2 = \sigma_{n,1}^2 - \theta_n = \Theta\left(a^{2n}\right)$,  $\lambdahat = \Theta(1)$ and $\lambda^\star_n = \Theta(1)$. This implies that the summands in~\eqref{eqn:Modified} for $i=1$ are both of the order $O(1/n)$ for any $x_1^2 = O(a^{4n})$. For $2\leq i\leq n$, the condition~\eqref{eqn:cond1} and the variance perturbation result in Theorem~\ref{th:vp} imply that every summand in~\eqref{eqn:Modified} for $i\geq 2$ is in the order of $\eta_n$. Hence, ~\eqref{eqn:Modified} is in the order of $\eta_n$. Finally, combining~\eqref{mbarbound} and~\eqref{eqn:Modified} implies that conditioning on the conditions~\eqref{eqn:cond1} and~\eqref{eqn:cond2}, we conclude that~\eqref{eqn:cond3} is violated with probability at most $\Theta(1 / \log n)$. 
\qed

\subsection{Auxiliary Lemmas}
\label{app:AL}

\begin{lemma}[Lower bound on the probability of distortion balls]
\label{lemma:shell}
Fix $d\in (0,\dmax)$. For any $n$ large enough and any $\bfu\in \mathcal{T}(n, p)$ defined in Definition~\ref{def:TS} in Section~\ref{subsec:LossyAEPandPE} above, and $\gamma$ defined by
\begin{align}
\gamma \triangleq \frac{(\log n)^{B_4}}{n} \label{gammaFS}
\end{align}
for a constant $B_4 > 0$ specified in~\eqref{B4}, below, it holds that 
\begin{align}
\mathbb{P}\left[d - \gamma \leq \mathsf{d}\left(\bfx, \hat{F}_1^{\star n}\right)\leq d ~| \hat{\bfX} = \bfx\right] \geq \frac{K_1}{\sqrt{n}}, \label{eqn:lowerShell}
\end{align}
where $K_1>0$ is a constant and $\hat{F}_1^{\star n}$ is in Appendix~\ref{app:lossyAEPnotations} above.
\end{lemma}
The proof is in Appendix~\ref{app:pfLemmaShell}.

\begin{lemma}
\label{lemma:generalized}
Fix $d\in (0,\dmax)$ and $\epsilon\in (0,1)$. There exists constants $C$ and $K_2>0$ such that for all $n$ large enough, 
\begin{align}
& \prob{\Lambda_{Y_1^{\star n}}\left(\bfX, \lambdahat, d\right) \leq \Lambda_{Y_1^{\star n}}\left(\bfX, \lambda^\star_n, d\right)  + C\log n} \notag \\
\geq & 1 - \frac{K_2}{\sqrt{n}},
\end{align}
where $\lambda^\star_n$ and $\lambdahat$ are defined in~\eqref{lambdaOpt} and~\eqref{lambdaXOpt}, respectively.
\end{lemma}
\begin{proof}
The proof of Lemma~\ref{lemma:generalized} is the same as~\cite[Eq.~(314)--(333)]{dispersionJournal} except that we strengthen the right side of~\cite[Eq.~(322)]{dispersionJournal} to be $\Theta(e^{-cn})$ for a constant $c> \log(a) / 2$ due to Corollary~\ref{cor:disp}.
\end{proof}

\subsection{Proof of Lemma~\ref{lemma:LossyAEP}}
\label{app:LossyAEP}
Using Lemmas~\ref{lemma:shell} and~\ref{lemma:generalized} in Appendix~\ref{app:AL} above, the proof of Lemma~\ref{lemma:LossyAEP} is almost the same as that in the stationary case~\cite[Eq. (270)-(278)]{dispersionJournal}. For completeness, we sketch the proof here. We weaken the bound~\cite[Lem. 1]{kostina2012fixed} by setting $P_{\hat{X}}$ as $P_{\hat{X}_1^n}$ and $P_Y$ as $P_{Y_1^{\star n}}$ to obtain that for any $\bfx\in \mathbb{R}^n$, 
\begin{align}
& \log\frac{1}{P_{Y_1^{\star n}}\left (\mathcal{B}(\bfx, d)\right )}\leq  \inf_{\gamma > 0}\Bigg\{
\Lambda_{Y_1^{\star n}}(\bfx, \lambdahat, d) + \lambdahat n \gamma -  \notag  \\
&\quad \log \prob{d - \gamma \leq \dis{\bfx}{\hat{F}_1^{\star n} } \leq d| \hat{X}_1^n = \bfx} \Bigg\},\label{pfeqn:lemmaK}
\end{align}
where $\lambdahat$ in~\eqref{lambdaXOpt} depends on $X_1^n$. Let $\mathcal{E}$ denote the event  inside the square brackets in~\eqref{eqn:lossyAEP}. Then, 
\begin{align}
& \mathbb{P}[\mathcal{E}] \notag \\
=~& \mathbb{P}[\mathcal{E} \cap \mathcal{T}(n, p)] +  \mathbb{P}[\mathcal{E} \cap \mathcal{T}(n, p)^c] \\
\leq ~& \mathbb{P} \Big [ \Lambda_{Y_1^{\star n}}(\bfX, \lambdahat, d)  \geq \Lambda_{Y_1^{\star n}}(\bfX, \lambda^\star_n, d) + p(n) -  \lambdahat n \gamma- \notag \\
& \quad\quad \frac{1}{2}\log n + \log K_1,~\mathcal{T}(n, p) \Big ] +  \mathbb{P}[ \mathcal{T}(n, p)^c] \label{appEQN:lossy}\\
 \leq ~& \mathbb{P} \Big [ \Lambda_{Y_1^{\star n}}(\bfX, \lambdahat, d)  \geq \Lambda_{Y_1^{\star n}}(\bfX, \lambda^\star_n, d) + C\log n \Big ]  + \notag \\
&\quad\quad \mathbb{P}[ \mathcal{T}(n, p)^c] \label{appEQN:lossyL}\\
\leq ~& \frac{1}{q(n)}, \label{appEQN:lossyM}
\end{align}
where
\begin{itemize}
\item \eqref{appEQN:lossy} is due to~\eqref{pfeqn:lemmaK} and Lemma~\ref{lemma:shell};

\item From~\eqref{appEQN:lossy} to~\eqref{appEQN:lossyL}, we used the fact that for $\bfu\in \mathcal{T}(n, p)$, $\lambdahat$ can be bounded by 
\begin{align}
\lrabs{\lambdahat - \frac{1}{2\theta}}\leq B_1, \label{eqn:lambdax}
\end{align}
where $B_1>0$ is a constant and $\theta > 0$ is given by~\eqref{eqn:RWD}. The bound~\eqref{eqn:lambdax} is obtained by the same argument as that in the stationary case~\cite[Eq. (273)]{dispersionJournal};  $\gamma$ is chosen in~\eqref{gammaFS} above; the constants $c_i$'s, $i = 1,...4$ in~\eqref{eqn:pn} are chosen as 
\begin{align}
c_1 &=  B_1 + \frac{1}{2\theta},\\
c_2 &= B_4,\\
c_3 &= C+ \frac{1}{2},\\
c_4 &= -\log K_1,
\end{align}
where $B_4>0$ is given in~\eqref{B4} below, and $K_1$ and $C$ are the constants in Lemmas~\ref{lemma:shell} and~\ref{lemma:generalized}, respectively.

\item \eqref{appEQN:lossyM} is due to Lemma~\ref{lemma:generalized} and Theorem~\ref{thm:typical}. 
\end{itemize}
\qed

\subsection{Proof of Lemma~\ref{lemma:shell}}
\label{app:pfLemmaShell}
\begin{proof}
The proof is similar to the stationary case~\cite[Lem. 10]{dispersionJournal}. We streamline the proof and point out the differences. Conditioned on $\hat{X}_1^n = \bfx$, the random variable
\begin{align}
\mathsf{d}\left(\bfx, \hat{F}_1^{\star n}\right) = \frac{1}{n}\sum_{i = 1}^n \left(\hat{F}_i^\star - x_i\right)^2 \label{dsum}
\end{align}
follows a noncentral $\chi^2$-distribution with (at most) $n$ degrees of freedom, since it is shown in~\cite[Eq. (282) and Lem. 4]{dispersionJournal} that conditioned on $\hat{X}_1^n = \bfx$, the distribution of the random variable $\hat{F}_{i}^\star - x_i$ is given by 
\begin{align}
\mathcal{N}\left(\frac{-x_i}{1 +2 \hat{\lambda}_n \nu_{n,i}^2},~ \frac{\nu_{n,i}^2}{1 + 2 \hat{\lambda}_n \nu_{n,i}^2}\right), \label{pfeqn:summand}
\end{align}
where $\nu_{n,i}^2$'s are given in~\eqref{nui}. Then, the conditional expectation is given by 
\begin{align}
\EX{\mathsf{d}\left(\bfx, \hat{F}_1^{\star n}\right) | \hat{X}_1^n = \bfx} = \frac{1}{n}\sum_{i = 1}^n m_i(\bfu), 
\label{pfeqn:EXDXF}
\end{align}
where $m_i(\bfu)$ is defined in~\eqref{mi} in Section~\ref{subsec:dispersion} above. In view of~\eqref{dsum},~\eqref{pfeqn:EXDXF} and~\eqref{eqn:cond3}, we expect that $\mathsf{d}\left(\bfx, \hat{F}_1^{\star n}\right)$ concentrates around $d$ conditioned on $\hat{X}_1^n = \bfx$ for $\bfu\in\mathcal{T}(n,p)$. Note that the proof of Theorem~\ref{thm:typical} related to~\eqref{eqn:cond3} is different from the one in the stationary case, see Appendix~\ref{app:PfThTS} above for the details. To simplify notations, we denote the variances as
\begin{align}
V_i(\bfx) &\delequal  \text{Var}\left[\left(\hat{F}_i^\star - x_i\right)^2 | \hat{X}_1^n = \bfx\right], \\
V(\bfx) &\delequal \sqrt{\frac{1}{n}\sum_{i  = 1}^n V_i(\bfx)}.
\end{align}
Due to~\eqref{pfeqn:summand} and~\eqref{eqn:cond3}, we see $(\hat{F}_i^\star - x_i)^2$'s have finite second- and third- order absolute moments. That is, we have 
\begin{align}
V(\bfx) = \Theta(1), \label{eqn:VVX}
\end{align}
for $\bfu\in \mathcal{T}(n,p)$. Therefore, we can apply the Berry-Esseen theorem. Hence,  
\begin{align}
& \mathbb{P}\left[d - \gamma \leq \mathsf{d}\left(\bfx, \hat{F}_1^{\star n}\right)\leq d ~| \hat{X}_1^n = \bfx\right] \notag \\
=~ & \mathbb{P}\Bigg[\quad\quad \frac{n(d - \gamma) -\sum_{i = 1}^n m_i(\bfu)}{\sqrt{n} V(\bfx)}  \notag \\
&\quad\quad \leq   \frac{1}{\sqrt{n} V(\bfx)} \sum_{i = 1}^n \left[\left(\hat{F}_i^\star - x_i\right)^2 - m_i(\bfu)\right]  \notag \\
&\quad\quad \leq  \frac{nd  -\sum_{i = 1}^n m_i(\bfu)}{\sqrt{n} V(\bfx)} ~|~ \hat{X}_1^n = \bfx\quad \Bigg] \\
\geq~ & \Phi\left(\frac{nd  -\sum_{i = 1}^n m_i(\bfu)}{\sqrt{n} V(\bfx)}\right) \notag \\
& \quad -  \Phi\left(\frac{n(d-\gamma)  -\sum_{i = 1}^n m_i(\bfu)}{\sqrt{n} V(\bfx)}\right) - \frac{2B_1}{\sqrt{n}} \label{pfsteps:PHI1}\\
= ~&  \frac{\sqrt{n} \gamma}{V(\bfx)} \Phi'(\xi) - \frac{2B_1}{\sqrt{n}},\label{pfsteps:PHI2}
\end{align}
where 
\begin{itemize}
\item \eqref{pfsteps:PHI1} follows from the Berry-Esseen theorem;  $B_1 >0$ is a constant, and 
\begin{align}
\Phi(t) \delequal \frac{1}{\sqrt{2\pi}}\int_{-\infty}^t e^{-\frac{\tau^2}{2}}~d\tau
\end{align}
is the cumulative distribution function of the standard Gaussian distribution;
\item \eqref{pfsteps:PHI2} is due to the mean value theorem and 
\begin{align}
\Phi'(t) = \frac{1}{\sqrt{2\pi}}e^{-\frac{t^2}{2}};
\end{align}

\item In \eqref{pfsteps:PHI2}, $\xi$ satisfies 
\begin{align}
\frac{n(d-\gamma)  -\sum_{i = 1}^n m_i(\bfu)}{\sqrt{n} V(\bfx)} \leq \xi \leq \frac{nd  -\sum_{i = 1}^n m_i(\bfu)}{\sqrt{n} V(\bfx)}. \label{xixixi}
\end{align}
\end{itemize}

By~\eqref{eqn:cond3} and~\eqref{eqn:VVX}, we see that there is a constant $B_2 > 0$ such that 
\begin{align}
\lrabs{\frac{nd - \sum_{i = 1}^n m_i(\bfu)}{\sqrt{n} V(\bfx)}} \leq B_2\sqrt{\log\log n}.
\end{align}
Hence, as long as $\gamma$ in~\eqref{xixixi} satisfies
\begin{align}
\gamma \leq O(\eta_n), \label{cond:gamma}
\end{align}
where $\eta_n$ is defined in~\eqref{eqn:etan}, there exists a constant $B_3>0$ such that 
\begin{align}
|\xi| \leq B_3 \sqrt{\log\log n}. \label{xixi}
\end{align} 
Let $B_4 > 0$ be a constant such that 
\begin{align}
B_4 \geq \frac{B_3^2}{2} + 1,\label{B4}
\end{align}
and choose $\gamma$ as in~\eqref{gammaFS}, which satisfies~\eqref{cond:gamma}. Then, plugging the bounds~\eqref{eqn:VVX},~\eqref{xixi},~\eqref{B4} and~\eqref{gammaFS} into~\eqref{pfsteps:PHI2}, we conclude that there exists a constant $K_1>0$ such that~\eqref{pfsteps:PHI2} is further bounded from below by $\frac{K_1}{\sqrt{n}}$.
\end{proof}

\bibliographystyle{IEEEtran}
{\small
\bibliography{main}}

\begin{IEEEbiographynophoto}{Peida Tian}(S'18)
is a Ph.D candidate in the Department of Electrical Engineering at California Institute of Technology. He received a B. Engg. in Information Engineering and a B. Sc. in Mathematics from the Chinese University of Hong Kong (2016), and a M.S. in Electrical Engineering from Caltech (2017). He is interested in optimization and information theory.
\end{IEEEbiographynophoto}

\begin{IEEEbiographynophoto}{Victoria Kostina}(S'12--M'14)
received a bachelor's degree from Moscow Institute of Physics and Technology (2004), where she was affiliated with the Institute for Information Transmission Problems of the Russian Academy of Sciences, a master's degree from University of Ottawa (2006), and a PhD from Princeton University (2013). She joined Caltech in the fall of 2014, where she is currently a Professor of Electrical Engineering. She received the Natural Sciences and Engineering Research Council of Canada postgraduate scholarship (2009--2012), the Princeton Electrical Engineering Best Dissertation Award (2013), the Simons-Berkeley research fellowship (2015) and the NSF CAREER award (2017).  Kostina's research spans information theory, coding, control, learning, and communications. 
\end{IEEEbiographynophoto}

\end{document}